\documentclass{article}
\usepackage[top=1.in, bottom=1.in, left=1.in, right=1.in]{geometry}
\usepackage{graphicx} 
\usepackage{subcaption}
\usepackage{amsmath,amssymb,amsfonts,amsthm,bbm} 
\usepackage{xcolor}
\usepackage[utf8]{inputenc}
\usepackage{algorithm}
\usepackage[noend]{algpseudocode} 
\usepackage[numbers]{natbib}

\usepackage{array,makecell,tabularx}

\usepackage{booktabs}

\usepackage{amsmath,amsfonts,bm}

\def\eqref#1{equation~\ref{#1}}

\def\1{\bm{1}}

\def\vone{{\bm{1}}}

\def\vp{{\bm{p}}}
\def\vq{{\bm{q}}}

\def\mQ{{\bm{Q}}}

\DeclareMathAlphabet{\mathsfit}{\encodingdefault}{\sfdefault}{m}{sl}
\SetMathAlphabet{\mathsfit}{bold}{\encodingdefault}{\sfdefault}{bx}{n}

\newcommand{\E}{\mathbb{E}}

\newcommand{\R}{\mathbb{R}}

\DeclareMathOperator*{\argmax}{arg\,max}
\DeclareMathOperator*{\argmin}{arg\,min}

\usepackage{hyperref}
\usepackage{thmtools}
\usepackage[capitalise]{cleveref}

\newtheorem{theorem}{Theorem}

\theoremstyle{definition}
\newtheorem{definition}{Definition}
\newtheorem{remark}{Remark}

\theoremstyle{plain}

\usepackage{comment}

\crefname{appendix}{Appendix}{Appendices}
\Crefname{appendix}{Appendix}{Appendices}

\newcommand{\val}{\operatorname{val}}
\newcommand{\pay}{\operatorname{pay}}
\newcommand{\pen}{\operatorname{pen}}

\DeclareMathOperator{\equal}{equalize}
\newcommand{\truthidx}{{i_\text{truth}}}
\newcommand{\FVal}{\texttt{ComputeVal}}

\title{Optimally Auditing Adversarial Agents}

\author{
  Sanmay Das\thanks{Authors listed in alphabetical order.} \\
  Virginia Tech \\
  \texttt{sanmay@vt.edu}
  \and
  Fang-Yi Yu\footnotemark[1] \\
  George Mason University \\
  \texttt{fangyiyu@gmu.edu}
  \and 
  Yuang Zhang\footnotemark[1] \\
  George Mason University \\
  \texttt{yzhang78@gmu.edu}  
}

\date{}

\begin{document}

\maketitle

\begin{abstract}
    Fraud can pose a challenge in many resource allocation domains, including social service delivery and credit provision. For example, agents may misreport private information in order to gain benefits or access to credit. To mitigate this, a principal can design strategic audits to verify claims and penalize misreporting. In this paper, we introduce a general model of audit policy design as a principal-agent game with multiple agents, where the principal commits to an audit policy, and agents collectively choose an equilibrium that minimizes the principal’s utility. We examine both adaptive and non-adaptive settings, depending on whether the principal's policy can be responsive to the distribution of agent reports. Our work provides efficient algorithms for computing optimal audit policies in both settings and extends these results to a setting with limited audit budgets.

\end{abstract}
\section{Introduction}

AI is increasingly used in making high-stakes societal decisions. One example that has recently gained considerable attention is the use of AI to decide whether to approve or deny the receipt of social benefits, with worries about how the scale of AI might systematically cut off thousands of people from benefits they are eligible for because of suspicion of fraud~\cite{eubanks2018automating}. The reason to use AI in these domains, is because human time and resources are limited. However, an alternative method is the use of AI to flag a limited number of applicants for \emph{audits} that can then be conducted by humans. How should one design such audit policies?

The problem of optimal audit design is relevant not just in the case of benefit receipt, but also in many other scenarios where agents must report their types to a principal in order for the principal to decide whether or not an agent is qualified to receive a benefit or service from the principal. In addition to qualification to receive social services or government benefits, other examples include credit or loan applications and tax relief. In all of these, the principal has the ability to, at some cost, audit agents to determine whether or not they are revealing their true types. The principal may also be able to impose a penalty on those caught misreporting (e.g., prosecution of tax fraud, ineligibility for future government services). The goals of the principal can vary. In some cases, they may want to minimize misreporting. For example, a social services agency can be thought of as a benevolent principal trying to maximize the social welfare of recipients, and it needs truthful elicitation in order to achieve this goal. In others, the principal may have their own utility function -- for example, a bank making lending decisions. 

There have been a number of papers that look at specific versions of the general auditing problem described above. In this paper, we systematically analyze the problem along three different dimensions, and present a number of new results. The first dimension is the goal, distinguishing between maximizing social welfare over the agents and maximizing the principal's utility. The second is whether or not the principal's strategy can be responsive to the actual distribution of agent reports (we call these the adaptive versus non-adaptive settings). Finally, the third dimension is whether audits are limited by a budget or whether the principal can undertake marginal-cost auditing by paying a specific cost for each additional audit. 

To give a concrete example, consider the Social Security Administration's Supplemental Security Income program, which has a strict upper limit on assets for eligibility.
The U.S.\ Supplemental Security Income program provides monthly cash only to applicants who (among other requirements) claim less than \$ 2000 in so-called ``deemed'' resources (essentially countable assets). As verifying every claim is costly, the agency can only audit a fraction of applicants. Deciding which brackets to inspect, whether to adapt those fractions after seeing the week’s claims, and how to weigh extra recoveries against audit costs mirrors the three axes we analyze: objective (social welfare vs.~agency pay‑off), adaptivity (fixed vs.~responsive rules), and resources (budget vs.~cost).
Our model captures this core strategic tension: applicants choose whether to misreport their private asset levels, anticipating the mixed audit policy chosen by the principal.

\paragraph{Our contributions}
This paper studies audit mechanism problems when agents adversarially choose the \emph{worst equilibrium}.  A principal wants truthful reports from $n$ strategic agents with private types in $[m]$. She may commit in advance to a probability of auditing different types (audit vector)
or to an adaptive policy that decides the audit vector after seeing the agents' reports. Each audit costs $\lambda$ (in the costly setting) or counts against a budget of $B$ (in the budgeted setting); detected misreports incur a penalty.  After the principal commits, agents adversarially select the worst-case equilibrium.
\begin{enumerate}
    \item We fully characterize the equilibrium structure in the non-adaptive costly setting where the principal commits to an audit vector and each audit costs~$\lambda$.  This structure yields an $\epsilon$-approximation algorithm in $O(m^2)$ time for the principal’s utility (Theorem~\ref{thm:opt_vector}).  We further prove that exact optimality is impossible in~\Cref{prop:nomax}.
    \item When the prior is unknown and varies in each round, we develop an online learning algorithm in \S~\ref{sec:noregret} that has regret $O(n\sqrt{Tm^2\log m})$ in $T$ rounds (\Cref{thm:noregret}).  Interestingly, although exact optimality in the one-shot setting is impossible, careful choice of arms allows our online learning algorithm to satisfy the no regret property.
    \item Beyond the principal’s utility, \S\ref{sec:non-adaptive_welfare} adapts both the efficient algorithm and the online learner to maximize social welfare, and \S\ref{sec:non-adaptive_penality} shows that increasing the penalty function or decreasing audit cost can only benefit the principal's utility and social welfare (\Cref{prop:monotone}).
    \item For adaptive audits, although the principal has a larger action space, we show that they offer no advantage over non-adaptive audits under the insensitivity assumption~\cref{eq:assum4} and the Wardrop equilibrium~\cref{eq:eqi_adaptive}. The same algorithm applies in the adaptive costly-audit setting (\Cref{thm:opt_adaptive_costly}). A similar algorithm for the budgeted case appears in~\cref{sec:adaptive_budgeted}.
\end{enumerate}
\subsection{Related Work} 
At a high level, our setting is a principal--(multi‑)agent Stackelberg game. In this section, we survey relevant techniques and highlight connections to three special cases---audit games, security games, and toll pricing in congestion games.

Our Stackelberg game features one leader and multiple heterogeneous followers.  Computing the leader’s optimal commitment is known to be hard even with two followers~\citep{conitzer2006computing}.  Our robustness notion relates to standard pessimistic equilibria~\citep{coniglio2017pessimistic}, but as our model has a larger action space (real-valued probability of auditing each report) with a non-convex structure, the standard bi-level optimization technique is not feasible. Recent work also considers solving multi-follower games under no externality assumptions~\citep{personnat2025learningplaymultifollowerbayesian}.

Classic costly-state-verification work \citep{townsend1979optimal} and subsequent allocation papers \citep{mookherjee1989optimal,ben2014optimal} study auditing or verification in resource-allocation settings, typically targeting truthful outcomes.   \citet{alm1993audit,ben2015auditing,coates2002financial} analyze equilibria of audit games, whereas our work addresses the mechanism design problem. 
\citet{10.1145/3328526.3329623} study penalty design with an exogenous audit process, whereas in our setting the principal designs the audit strategy. 
In multi-agent settings, \citet{Estornell2020StrategicClassification,estornell2023incentivizing} use audits to discourage misreporting and promote beneficial recourse.  A key difference between our work and the above is equilibrium selection: rather than targeting truth‑telling equilibria, we guarantee the principal’s performance at the worst (pessimistic) equilibrium, which may be non‑truthful. Recently, \citet{Jalota2024CatchFraud} connect information design to audit mechanisms when the agents can commit to a misreporting strategy; in our work, the principal is the one who can commit. 


Less directly related are security games, where a defender allocates inspection or patrol resources to deter attackers and inspection costs do not scale with the number of attackers~\citep{pita2008deployed,tambe2011security}. Closer to us are audit games that allow the leader to tune punishment for a single agent~\citep{blocki2013auditgames,blocki2015auditgamesmultipledefender}.

Our audit probabilities play a role analogous to tolls in congestion games: they modify followers’ payoffs to steer equilibrium flows.  However in our model an agent's cost depends not only on its reported type but also the true type.  Foundational work showed that marginal‐cost tolls can implement the system optimum in nonatomic traffic ~\citep{roughgarden2002how,cole2003pricing}.







\section{Audit Mechanism Problem}\label{sec:pre_model}
We study an audit mechanism design problem where a principal interacts with a continuum of (non-atomic) agents with total mass $n$. Each agent has a private type drawn from a known prior and may choose to misreport this type. The principal aims to incentivize truthful reporting by combining costly audits and penalties.  We extend our analysis to a setting with a hard audit-budget constraint in~\cref{sec:adaptive_budgeted}.

\paragraph{Basic setup}
There are $m\ge 2$ ordered types, denoted by $[m] = \{0,1,\dots,m-1\}$.  Each agent has a private type $i\in [m]$ which is sampled independently from a prior $\vq\in \Delta_m$ with full support on $[m]$, and reports $k\in [m]$ which may differ from $i$. We use $i,j$ for true types and $k,l$ for reported types.

When a type $i$ agent reports $k$, the principal assigns a payment $\pay(k)$, and receives $\val(i,k)$. The principal detects misreporting through audits by choosing an \emph{audit vector} $\vp\in [0,1]^m$  where $p_k$ is the probability of auditing an agent reporting type $k$.  Once audited, the principal gets $\pen(i,k)$ from the agent. We consider a type-independent penalty of the form $\pen(k)$ so that for all $i,k$
$\pen(i,k) = \mathbf{1}[i\neq k]\cdot\pen(k).$%
\footnote{A type-independent penalty can admit a weaker notion of auditing---one that can detect inconsistencies between the reported and true types but cannot identify the true type itself.}
One example is an affine penalty where $\pen(k) = a\pay(k) + b$ with $a, b \ge 0$.\footnote{%
Here are two real-world examples of affine penalties.  China’s Export Control authority levies fines between five and ten times the illicit turnover from an unlicensed export, i.e., $\pen(k) = 10\pay(k)$ or $\pen(k) = 5\pay(k)$.  Virginia requires an evading driver to pay the unpaid toll and an administrative fee of up to \$100, i.e., $\pen = \pay+100$.} This includes the formulation in~\cite{Estornell2020StrategicClassification} as a special case when $a = 1$.

Without loss of generality, we order the indices so that $0<\pay(k)<\pay(l)$ for all $k< l$ and $\pay(-1):=0$.  Additionally, we assume  misreporting higher can only decrease the value to the principal,
\begin{equation}\label{eq:assum2}
    \val(i,k)\ge \val(i,l)\text{ for all }i\le k \le l,
\end{equation}
and an agent that knows it will be audited would have no incentive to misreport: 
\begin{equation}\label{eq:assum1}
    \pen(k)\ge \pay(k)\text{ for all }k.
\end{equation} 

\paragraph{Agents' utilities and strategies}
Given an audit vector $\vp$, a type $i$ agent reporting $k$ has expected utility
\begin{equation}\label{eq:util_agent}
    U_{i,k}(\vp):= \pay(k)-p_k\pen(i,k)
\end{equation}

Agents use a randomized \emph{report strategy} represented by a matrix $\mQ$ where $Q_{i,k}$ is the probability of a type $i$ agent reporting type $k$. The induced \emph{report distribution}\footnote{As agents are non-atomic, the observed report distribution equals the expectation.  In particular, if all are truthful, the report distribution equals $\vq$.} is $\hat{\vq}\in \Delta_m$ where $\hat{q}_k = \sum_i q_i Q_{i,k}$ for all $k$.  
\begin{definition}\label{def:bne1}
      Given $\vp$, a report strategy $\mQ$ is a \emph{Bayes-Nash equilibrium} if for all $i$ and $k, l\in [m]$ with $Q_{i,k}>0$, $$U_{i,k}(\vp)\ge U_{i,l}(\vp).$$
Let $Eqi(\vp)$ be the set of all equilibria. 
\end{definition}

\paragraph{Principal's utility}
Given $\vp$ and $\mQ$, let 
$$C(\vp,\mQ) := n\sum_{i,k}q_iQ_{i,k}p_k$$ be the expected number of audits, and the principal’s utility without audit costs 

$$V(\vp, \mQ) = n\sum_{i, k\in [m]} q_iQ_{i,k}\left(\val(i,k)-\pay(k)+p_k\pen(i,k)\right)$$
where the final term is the gain from auditing.

We consider the \emph{costly setting} where the principal can audit any number of agents, but incurs a cost $\lambda\ge 0$ per audit. The principal's utility is 
\begin{equation}\label{eq:util_principal}
    V_\lambda(\vp, \mQ) = V(\vp,\mQ)-\lambda C(\vp,\mQ)
\end{equation}
We assume that the cost of audits is less than the penalty
\begin{equation}\label{eq:assum3}
    \lambda\le \pen(k)\text{ for all }k.
\end{equation}
We defer the budgeted setting to~\cref{sec:adaptive_budgeted}.
\paragraph{Principal strategies: Non-adaptive and adaptive}
The principal’s audit vector may be fixed or adaptively chosen based on agents' reports, $\hat{\vq}$. 

In the \emph{non-adaptive} setting, the principal commits to an audit vector $\vp$. After observing $\vp$, the agents collectively choose an equilibrium $\mQ\in Eqi(\vp)$ that is worst for the principal.
In the \emph{adaptive} setting, the principal commits to an \emph{audit strategy} $\pi$, which maps a reported distribution $\hat{\vq}$ to an audit vector $\vp = \pi(\hat{\vq})$.  After observing $\pi$, all agents collectively choose a worst equilibrium $\mQ$ under $\pi$ so that 
\begin{equation}\label{eq:eqi_adaptive}
    U_{i,k}(\pi(\hat{\vq}))\ge U_{i,l}(\pi(\hat{\vq})),\forall i, k, l\in [m]\text{ with }Q_{i,k}>0
\end{equation}
The above is a \emph{Wardrop equilibrium}; a single agent's deviation does not change the report distribution $\hat{\vq}$.
We denote $Eqi(\pi)$ as the set of equilibria among agents under strategy $\pi$, and 
$V_\lambda(\pi,\mQ) = V_\lambda(\pi(\hat{\vq}), \mQ)$
where $\hat{\vq}$ is the report distribution of $\mQ$.

\section{Optimal Non-Adaptive Audits with Costs}\label{sec:non-adaptive}
We study non-adaptive costly audit games with $(n, m, \vq, \val, \pay, \pen)$ and $\lambda$. Most proofs are deferred to~\cref{app:non-adaptive}.

\subsection{Optimizing the principal's utility}\label{sec:non-adaptive_utility}
The principal wants to maximize her utility under the worst-case Bayes-Nash equilibrium, defined below
\begin{equation}\label{eq:utiliy_max}
    V_\lambda(\vp):=\min_{\mQ \in Eqi(\vp)} V_\lambda(\vp, \mQ).
\end{equation}
An audit vector $\vp$ \emph{$\epsilon$-approximates} $\vp'$ if 
$V_\lambda(\vp) \ge V_\lambda(\vp') -\epsilon$, and $\vp$ is \emph{$\epsilon$-optimal} if it $\epsilon$-approximates any $\vp'$, i.e. $V_\lambda(\vp) \ge \sup_{\vp'}V_\lambda(\vp') -\epsilon$.

\Cref{thm:opt_vector} shows that there exists an algorithm that computes an $\epsilon$-optimal audit vector in time $O(m^2)$.  This runtime is tight, as reading all entries of $\val$ already requires $\Omega(m^2)$ time. Moreover, \Cref{prop:nomax} shows that computing an exactly optimal audit vector is impossible.

\begin{theorem}[Utility-optimal]\label{thm:opt_vector}
For any small enough $\epsilon>0$, $(n$, $m$, $\vq$, $\val$, $\pay$, $\pen)$ and $\lambda$, Algorithm~\ref{alg:water_filling_algorithm} computes a $2n\epsilon$-optimal audit vector for \cref{eq:utiliy_max} in $O\allowbreak(m^4)$ time.  

Moreover, the time complexity can be improved to $O(m^2)$.
\end{theorem}

The idea of \Cref{alg:water_filling_algorithm} is to search over a finite set of audit vectors, called critical audit vectors (\Cref{def:equal}). We also  show that any audit vector can be  approximated by one from this set.

The remainder of this section is organized as follows. We begin by defining equalized and critical audit vectors and presenting the algorithm. Next, we characterize agents' best responses and equilibrium behavior in \Cref{lem:br}, a result that underpins both \Cref{thm:opt_vector} and later analyses. We then show that exact optimization in \cref{eq:utiliy_max} may be impossible, justifying our approximation approach. Finally, we prove \Cref{thm:opt_vector}.

Let $\rho_k(u)$ be the probability that a type $k$ report is audited when $u$ is the utility of misreporting.
\begin{equation}\label{eq:rho}
    \rho_k(u) = \frac{\pay(k)-u}{\pen(k)}
\end{equation}
This is a valid probability when $0\le u\le \pay(k)$ by \cref{eq:assum1}.  Note that $\rho_k(u)$ is decreasing in $u$, and $p_k = \rho_k({U}_{i,k}(\vp))$, for all $\vp$ and $i\neq k$.  Hence, $\rho_k$ is a bijection between misreport utility and audit probability of type $k$.

\begin{definition}[Equalized and critical audit vectors]\label{def:equal}
    Given $0< u\le \max_k \pay(k)$ with $\iota = \min\{i: \pay(i)\ge u\}$, ${A}\subseteq \{i\in [m]: i\ge \iota\}$, and $0<\epsilon<u$, we define the \emph{equalized audit vector} $\vp=\equal(u, A,\epsilon)$ such that for all $k\in [m]$
\[p_k = 
\begin{cases}
0 & k< \iota,\\
\rho_{k}(u) & k\in {A},\\
\rho_{k}(u-\varepsilon) & \text{otherwise}.
\end{cases}
\]
If $\hat{A} = \{\kappa\}$, we write $\equal(u, A,\epsilon)$ as $\equal(u,\kappa,\epsilon)$.
Given $0<\epsilon< \gamma:=\min_k (\pay(k)-\pay(k-1))$ and $\iota\le \kappa$ we define the \emph{critical audit vectors} as
$\equal^+(\iota, \kappa,\epsilon)=\equal(\pay(\iota-1)+\epsilon,\kappa,\epsilon)$ and $\equal^-(\iota,\kappa,\epsilon)=\equal(\pay(\iota)-\epsilon,\kappa,\epsilon)$.
\end{definition}
Intuitively, an equalized audit vector sets the misreport value of all types in $A$ to $u$, and minimizes audit probabilities for others so that agents either misreport as $A$ or report truthfully. \Cref{lem:equal_welldefined} formalizes this property.  Note that because the $\rho_k$ are decreasing, $\equal^+(\iota, \kappa,\epsilon)\ge \equal^-(\iota, \kappa,\epsilon)$ coordinate-wise.  

\begin{algorithm}[t]
\caption{SuccinctSearch}
\label{alg:water_filling_algorithm}
\begin{algorithmic}[1]
\Require $\epsilon > 0$, $(n, m, \vq, \val, \pay, \pen)$, and $\lambda \ge 0$
\Ensure Audit vector $\vp^*$
\State Initialize $V_{\text{max}} \gets -\infty$ and $\vp^* \gets \vone$
\For{$i \in [m]$}
    \For{$k = i$ to $m - 1$}
        \State $\vp^+ \gets \equal^+(i, k, \epsilon)$ \Comment{critical audit vector}
        \State $\vp^- \gets \equal^-(i, k, \epsilon)$
        \If{$V_{\text{max}} < \Call{ComputeVal}{\vp^+}$}
            \State $\vp^* \gets \vp^+$
            \State $V_{\text{max}} \gets \Call{ComputeVal}{\vp^+}$
        \EndIf
        \If{$V_{\text{max}} < \Call{ComputeVal}{\vp^-}$}
            \State $\vp^* \gets \vp^-$
            \State $V_{\text{max}} \gets \Call{ComputeVal}{\vp^-}$
        \EndIf
    \EndFor
\EndFor
\State \Return $\vp^*$

\Function{ComputeVal}{$\vp$}
    \State Set $\hat{u} \gets \max_k \left\{ \pay(k) - p_k \pen(k) \right\}$, $\hat{A} \gets \arg\max_k \left\{ \pay(k) - p_k \pen(k) \right\}$, and $v \gets 0$.
    \For{$i \in [m]$}
        \State $v_i\gets \left( \val(i,i) - \pay(i) - p_i \lambda \right)$
        \State $\hat{v}_i\gets \min_{k \in \hat{A}} \left( \val(i,k) - \pay(k) + p_k (\pen(k) - \lambda) \right)$
        \If{$\pay(i) > \hat{u}$} \Comment{truthful}
            \State $v \gets v + q_i v_i$
        \ElsIf{$\pay(i) < \hat{u}$} \Comment{misreporting}
            \State $v \gets v + q_i \hat{v}_i$
        \Else \Comment{indifferent}
            \State $v \gets v + q_i \min \left\{ v_i,\hat{v}_i\right\}$
        \EndIf
    \EndFor
    \State \Return $v$
\EndFunction
\end{algorithmic}
\end{algorithm}

\subsubsection{Characterizing best response and equilibrium}
Before proving the theorem, we show that the best response of each agent follows a threshold structure.  There exists a minimal truthful type and a misreporting range such that all agents with lower types strictly prefer to misreport as a type within the misreport range, while higher types strictly prefer to report truthfully.

Given $\vp$, by \cref{eq:util_agent}, we can write the \emph{best-response set} of type-$i$ agents as 
$A_i(\vp) = \arg\max_{\,k\in [m]} (\pay(k) \;-\; p_k\,\pen(i,k)).$
To simplify the notation, we define the misreport utility of reporting $k$ as 
$$\hat{U}_{k}(\vp):=\pay(k)-p_k\pen(k)$$, which is independent of the misreporting agent's type, and the utility of being truthful as 
$U_k=\pay(k).$
Finally, let $\hat{u}(\vp) = \max_k \hat{U}_k(\vp)$ be the \emph{highest misreport utility}, $\truthidx(\vp) = \min\{i\in [m]: U_i\ge \hat{u}(\vp)\}$ be the \emph{minimal truthful type} (the lowest type that is willing to be truthful), and \emph{misreporting range}  $\hat{A}(\vp) = \argmax_k\{\hat{U}_k(\vp)\}\subseteq [m]$ be the set of types that have the highest misreport utility.  We will omit $\vp$ when it is clear in context.

\begin{restatable}[Threshold structure]{lemma}{lembr}\label{lem:br}
      Given $\vp$ with $\hat{u}(\vp)$, $\hat{A}(\vp)$, and $\truthidx(\vp)$ defined above, $\hat{A}(\vp)\subseteq \{k\in [m]: k\ge \truthidx(\vp)\}$ and 
      $$A_i(\vp) = \begin{cases}
      \{i\}&\text{ if }i>\truthidx,\\
          \hat{A}&\text{ if $i< \truthidx$}\\
          \{\truthidx\} &\text{ if $i= \truthidx$ and $U_{\truthidx}>\hat{u}$}\\
          \hat{A}\cup \{\truthidx\} &\text{ if $i= \truthidx$ and $U_{\truthidx}=\hat{u}$}
      \end{cases}.$$
\end{restatable}
An audit vector is \emph{strict} if $\hat{u}\notin\{U_i: i\in [m]\}$ so that every agent is either truthful or misreports as $\hat{A}$.  Additionally, a report strategy $\mQ$ is \emph{single-minded} with $\iota\le \kappa$ if all types $i\ge \iota$ are truthful and all types $i<\iota$ report as $\kappa$.\footnote{If $\iota = 0$, everyone is truthful.}

By \Cref{lem:br}, $Eqi(\vp)$ is non-empty and closed, so the minimum in \cref{eq:utiliy_max} is well-defined.  However, \Cref{prop:nomax} shows that the maximum of \cref{eq:utiliy_max} does not always exist.
\begin{figure}[ht]
    \centering
    \includegraphics[width=0.5\columnwidth]{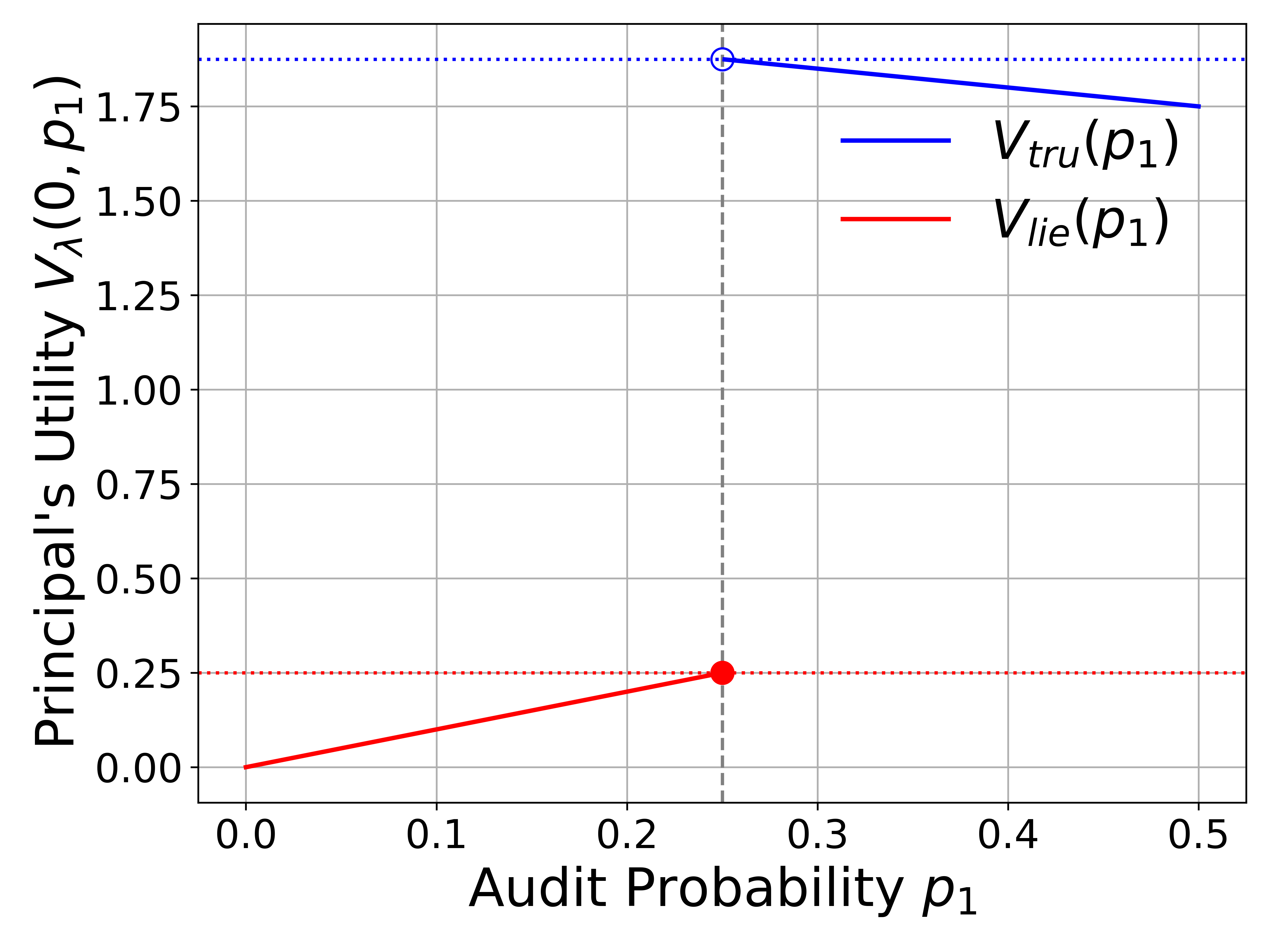}
    \caption{A non-adaptive audit game with unattainable optimum:  As in~\cref{prop:nomax}, we consider binary types ($m=2$) with $q_0 = q_1 = 1/2$, $\pay = (1,2)$, $\pen = (3,4)$, $\val = \begin{pmatrix}
3& 0\\
0& 4
\end{pmatrix}$, and $\lambda = 1$.  We vary the audit probability of the high type, using $\vp = (0, p_1)$ since auditing the low type is useless.  The principal’s utility if all agents misreport as the high type is $V_{lie}(p_1) = p_1$ (red), and if all are truthful is $V_{tru}(p_1) = \frac{4-p_1}{2}$ (blue). When $p_1 < \frac{1}{4}$, misreporting as the highest type is the unique equilibrium, for $p_1 > \frac{1}{4}$, truth-telling is the unique equilibrium, and at the threshold $p_1 = \frac{1}{4}$, any mixture is an equilibrium.  Therefore, the principal's worst-case equilibrium utility at $p_1 = \frac{1}{4}$ is $\frac{1}{4}$, but $\sup_\vp V_\lambda(\vp) = \lim_{p_1\to 1/4^+} V_\lambda(0,p_1) = \frac{15}{8}$ which is not attained by any $\vp$.}
    \label{fig:nomax}
\end{figure}
\begin{restatable}{proposition}{propnomax}\label{prop:nomax}
There exists a non-adaptive costly audit game so that $\sup_{\vp\in [0,1]^m}V_\lambda(\vp)<+\infty$ but  for all $\vp$,
$V_\lambda(\vp)<\sup_{\vp'\in [0,1]^m}V_\lambda(\vp').$
\end{restatable}

\Cref{lem:equal_welldefined} establishes that the equalized audit vector is well-defined and corresponds to the $u$, $A$, and $\iota$ in \cref{lem:br}.

\begin{restatable}{lemma}{lemequalwelldefined}\label{lem:equal_welldefined}
    Given $\vp = \equal(u, A, \epsilon)$ with $\iota = \min\{i: \pay(i)\ge u\}$, if $u\notin \{U_i: i\in [m]\}$, $\vp$ is a strict audit vector, $\hat{u}(\vp) = u$, $\hat{A}(\vp) = {A}$, and $\truthidx(\vp) = \iota$.  
    $$A_i(\vp) = \begin{cases}
      \{i\}&\text{ if }i\ge \truthidx,\\
          A&\text{ if $i< \truthidx$}
      \end{cases}.$$
\end{restatable}

\subsubsection{Approximation by equalized and critical audit vectors}

We now show that any audit vector can be approximated by some strict equalized vector (\Cref{lem:equal_approx}), and some critical vector (\Cref{lem:critical_approx}).

\begin{restatable}{lemma}{lemequalapprox}\label{lem:equal_approx} 
For any $\vp\in [0,1]^m$, there is a strict equalized audit vector  $\vp' = \equal(u, \kappa, \epsilon)$ with $u> 0, \kappa\in [m]$, and $0<\epsilon<u$ so that $Eqi(\vp')\subseteq Eqi(\vp)$ and for all $\mQ\in Eqi(\vp')$ 
$$V_\lambda(\vp, \mQ)\le V_\lambda(\vp', \mQ)+n\epsilon.$$
\end{restatable}
To prove~\cref{lem:equal_approx}, we use \Cref{lem:equal_welldefined} to find a a strict equalized audit vector $\vp'$ so that $Eqi(\vp') = \{\mQ\}\subseteq Eqi(\vp)$ consists of a single-minded equilibrium with $\iota$ and $\kappa$.  Then we upper bound the difference of the principal's utilities with these two audit vectors, $V_\lambda(\vp, \mQ)-V_\lambda(\vp', \mQ)$ to get
\begin{equation}\label{eq:equal_approx0}
    n\sum_{i, k\in [m]} q_iQ_{i,k}\left(\pen(i,k)-\lambda\right)\left(p_k-p_k'\right).
\end{equation}
To minimize \cref{eq:equal_approx0}, consider two cases.  If $k\neq \kappa$, the smaller $p_k'$ yields a larger utility (which is the intuition of \cref{eq:rho}).  For type $\kappa$, the equalized audit vector ensures $p_\kappa' \approx p_\kappa$.  

\begin{restatable}{lemma}{lemcriticalapprox}\label{lem:critical_approx} 
For any $\epsilon'$ and $\vp = \equal(u, \kappa, \epsilon)$ with $\iota = \truthidx(\vp)$ and $\epsilon, \epsilon'<\frac{\gamma}{2}$, there are $\vp^- = \equal( U_\iota-\epsilon',\kappa, \epsilon)$ and $\vp^+ = \equal( U_{\iota-1}+\epsilon',\kappa, \epsilon)$ so that $Eqi(\vp^+) = Eqi(\vp^-) = Eqi(\vp)$ and 
$$V_\lambda(\vp,\mQ)\le \max\{V_\lambda(\vp^+,\mQ), V_\lambda(\vp^-,\mQ)\}+n\epsilon'$$ for all $\mQ\in Eqi(\vp)$.
\end{restatable}

To prove \cref{lem:critical_approx}, we note that fixing $\kappa\in [m]$ and $\epsilon>0$, an equalized audit vector $\vp = \equal(u, \kappa, \epsilon)$ is parameterized by a single parameter $u\in \R$.  Moreover, the audit probabilities are affine in $u$, so the principal's utility is affine in $u$ by \cref{eq:equal_approx0}.  Therefore, we can optimize the principal's utility using the extreme value of $u\in (U_{\iota-1}, U_\iota)$ by~\cref{lem:equal_welldefined}.


\paragraph{Proof of Theorem~\ref{thm:opt_vector}}
Algorithm~\ref{alg:water_filling_algorithm} iterates over critical audit vectors with all combinations of $i\le k$ and computes the principal's worst-case utility.  By \cref{lem:br}, \FVal ~computes $V_\lambda(\vp)$ by considering whether a type is truthful or misreporting as $\hat{A}$.  Therefore, \Cref{alg:water_filling_algorithm} returns the optimal critical audit vector.

For the approximation guarantee, given any audit vector $\vp\in [0,1]^m$, by \Cref{lem:equal_approx}, there exists a strict equalized audit vector $\vp' = \equal(u, \kappa, \epsilon)$ with $\iota = \truthidx(\vp')$ so that
$V_\lambda(\vp)\le V_\lambda(\vp')+n\epsilon$.
By \Cref{lem:critical_approx}, there exists a critical audit vector $\vp'' = \equal^+(\iota, \kappa, \epsilon)$ or $\equal^-(\iota, \kappa, \epsilon)$ with $\iota\le \kappa$ so that 
$V_\lambda(\vp')\le V_\lambda(\vp'')+n\epsilon$.
Therefore, there exists some critical audit vector $$V_\lambda(\vp)\le V_\lambda(\vp'')+2n\epsilon,$$
and the algorithm is $2n\epsilon$-optimal.

The algorithm searches through all $2(\binom{m}{2}+m) = O(m^2)$ critical vectors.  By \cref{lem:br}, \FVal computes $V_\lambda(\vp)$ by computing the worst report in $A_i(\vp)$ that minimizes the principal's utility for all $i$.  This takes $O(m^2)$ for each audit vector. Therefore, the time complexity is in $O(m^4)$.  We can improve the running time of \Cref{alg:water_filling_algorithm} to $O(m^2)$ using dynamic programming for \FVal.

\begin{remark}We can improve the running time of \Cref{alg:water_filling_algorithm} to $O(m^2)$ using dynamic programming for \FVal as shown in~\cref{alg:fast_search}.  Because the equilibrium of any critical audit vector $\vp = \equal^+(i, k, \epsilon)$ or $\equal^-(i, k, \epsilon)$ is single-minded where agents with type indices strictly less than $i$ misreport as type $k$, and agents with type indices greater than or equal to $i$ report truthfully, and the principal's utility~\cref{eq:util_principal} becomes
\[
V_\lambda(\vp) = n \left[ M_{i,k}-(\pay(k)-p_{k}(\pen(k) - \lambda)) F_{i} + G_{i} \right],
\]
where $M_{i,k} = \sum_{j < i} q_j \val(j,k)$ is the total valuation from all misreporting agents with types less than $i$, $F_i = \sum_{j < i} q_j$ is the fraction of misreporting agents, and $G_i = \sum_{j \ge i} q_j (\val(j,j) - \pay(j)-p_j\lambda)$ is the utility from truthful agents with types greater than or equal to $i$.  We can use dynamic programming to compute the prefix sums $M_{i,k}, F_i$, and suffix sums $G_i$ for each $\iota \in [m]$, using $O(m^2)$ time.  Thus, we can evaluate the utility for any critical audit vector in  constant time by table lookups and make the running time in $O(m^2)$.  We provide the formal algorithm in the appendix.
\end{remark}

\subsection{No-Regret Auditing without a Prior}\label{sec:noregret}
One limitation of \Cref{alg:water_filling_algorithm} is assuming access to the prior $\vq$.  We provide a no-regret online learning algorithm when the prior $\vq$ is unknown and can vary in each round.

Let $V_\lambda(\vp, \mQ;\vq)$ be the principal’s (single-round) utility from \cref{eq:util_principal} under prior $\vq$, and $V_\lambda(\vp; {\vq}) := \min_{\mQ\in Eqi(\vp)}V_\lambda(\vp, \mQ; \vq)$.

Consider the principal and agents interacting over $T$ rounds. 
The principal knows $(n, m, \val, \pay, \pen)$ and $\lambda$ while Nature secretly chooses $\vec{\vq}:= (\vq^0,\dots,\vq^{T-1})$.  For round $t = 0,\dots,T-1$, 
\begin{enumerate}
    \item The principal with algorithm $\mathcal{A}$ samples an audit vector $\vp^t$ from a distribution $P^t$ based on the history $(\vp^0,v^0,\dots,\vp^{t-1},v^{t-1})$.
    \item After observing $\vp^t$, agents collectively choose the \emph{worst equilibrium} $\mQ^t\in \argmin_{\mQ\in Eqi(\vp^t)} V_\lambda(\vp^t, \mQ; \vq^t)$
    \item The principal gets $v^t = V_\lambda(\vp^t, \mQ^t; \vq^t) = V_\lambda(\vp^t;\vq^t)$.
\end{enumerate}
The principal designs an online learning algorithm $\mathcal{A}$ that maximizes her accumulative expected utility. 
Formally, the algorithm is evaluated by its (multi-agent) \emph{Stackelberg regret}~\cite{DBLP:journals/corr/abs-1710-07887,DBLP:journals/corr/abs-1911-04004}\footnote{Classical online‐Stackelberg work assumes a single agent (follower) who best-responds to the leader’s action. In our model the follower is a population of $n$ agents who play the worst equilibrium under the $n$-player game induced by the audit vector.}   against the optimal audit vector in hindsight which knows agents' prior $\vec{\vq}$.  We define \begin{equation}\label{eq:regret}
   Reg_T(\mathcal{A}, \vec{\vq}) = \sup_{\vp} \sum_{t\in [T]} V_\lambda(\vp; \vq^t)-\E_{\mathcal{A}}\left[\sum_{t\in [T]} V_\lambda(\vp^t; \vq^t)\right], 
\end{equation}
and $Reg_T(\mathcal{A}) = \sup_{\vec{\vq}} Reg_T(\mathcal{A}, \vec{\vq})$ 
where the randomness is over the choice of audit vectors.

\begin{restatable}{theorem}{thmnoregret}\label{thm:noregret}
Given any $(n,m,\val,\pay,\pen)$ and $\lambda$, the online learning algorithm $\mathcal{A}$ in \Cref{alg:exp3} has 
    $Reg_T(\mathcal{A}) = O(n\sqrt{Tm^2\log m}).$
\end{restatable}

The key observation is that the equalized and critical audit vectors in \Cref{def:equal} are independent of prior $\vq^t$.  Additionally, we can reuse~\Cref{lem:equal_approx,lem:critical_approx} to show that the critical vectors are approximately optimal as in~\Cref{lem:critical_noregret}.  

\begin{restatable}{lemma}{lemcriticalnoregret}\label{lem:critical_noregret}
Given any $0<\epsilon< \frac{\gamma}{2}$, there exist $\vp^+ = \equal^+(\iota, \kappa, \epsilon)$ or $\vp^-= \equal^-(\iota, \kappa, \epsilon)$ with $\iota\le \kappa$ so that for all $\vec{\vq} = (\vq^0,\dots,\vq^{T-1})$ and $\vp$,
$$V_\lambda(\vp; \vec{\vq})\le \max\{V_\lambda(\vp^+; \vec{\vq}), V_\lambda(\vp^-; \vec{\vq})\}+2n\epsilon T$$
where $V_\lambda(\vp; \vec{\vq}) := \sum_t V_\lambda(\vp; \vq^t)$.
\end{restatable}

With \Cref{lem:critical_noregret}, given $\epsilon>0$ we run a no regret algorithm for adversarial bandits (e.g., EXP3) on $O(m^2)$ critical audit vectors in order to achieve regret bounded by $O(\sqrt{Tm^2\log m}+n\epsilon T)$.  However, to achieve no-regret, \Cref{alg:exp3} considers the set of critical vectors in \Cref{def:equal} with $\epsilon_t = 2^{-t}\epsilon_0$.  Specifically, we consider the set of all critical audit vectors $\sigma \in \Sigma := \{ (i,k,+), (i,k,-) : i,k \in [m] \}$, and at round $t$ we use $\equal(\sigma,\epsilon_t) = \equal^+(i,k,\epsilon_t)$ if $\sigma = (i,k,+)$ and $\equal^-(i,k,\epsilon_t)$ if $\sigma = (i,k,-)$ as the set of arms.\footnote{We treat each tuple $(i,k,+)$ or $(i,k,-)$ as a template arm. EXP3 maintains weights over these templates, while the audit vector played in round $t$ depends on the template and $\epsilon_t$}

\begin{algorithm}[ht]
\caption{EXP3 algorithm on critical audit vectors}
\label{alg:exp3}
\begin{algorithmic}[1]
\Require Game parameters $(n, m, \val, \pay, \pen)$ with $L = n \max_{i,k} \left( \val(i,k) + \pay(k) + \pen(k) \right)$, cost $\lambda \ge 0$, horizon $T$, and learning rate $\eta = \sqrt{ \frac{\log (2m^2)}{2m^2 T} }$
\State Initialize $\epsilon_0 \gets \frac{\gamma}{3}$ and $s^0_\sigma \gets 0$ for all $\sigma\in \Sigma$.
\For{$t = 1$ to $T$}
    \State Compute $P_t$ with $P_{t,\sigma} \propto \exp(\eta s^t_\sigma)$ for all $\sigma$
    \State Sample $\sigma^t \sim P_t$ and set $\vp^t = \equal(\sigma^t, \epsilon_t)$
    \State Observe reward $v^t = V_\lambda(\vp^t; \vq^t)$,
    \State Update $\epsilon_{t+1} \gets \frac{1}{2} \epsilon_t$ and for all $\sigma$ $$s^{t+1}_\sigma \gets s^t_\sigma + 1 - \frac{L - v^t}{2L} \cdot \frac{\mathbb{I}[\sigma = \sigma^t]}{P_{t,\sigma}}.$$
\EndFor
\end{algorithmic}
\end{algorithm}


\subsection{Optimizing Social Welfare}\label{sec:non-adaptive_welfare}
Now we show how to maximize social welfare, the sum of the utilities of the principal and all agents, 
\begin{equation}\label{eq:welfare_max}
    \begin{aligned}
        W_\lambda(\vp, \mQ) := &V_\lambda(\vp, \mQ)+n\sum_{i,k\in [m]} q_iQ_{i,k}U_{i,k}(\vp)\\
        =& n\sum_{i,k\in [m]} q_iQ_{i,k}\left(\val(i,k)-p_{k}\lambda\right).
    \end{aligned}
\end{equation}
For instance, if $\val(i,k) = \mathbf{1}[i=k]$, the social welfare is the number of truthful agents minus the cost of audits.

As agents are strategic, we need to design an audit vector $\vp$ so that $W_\lambda(\vp):= \min_{\mQ\in Eqi(\vp)} W_\lambda(\vp, \mQ)$ is large, and say $\vp$ is $\epsilon$-optimal if $W_\lambda(\vp)\ge W_\lambda(\vp')-\epsilon$ for all $\vp'$.

\begin{theorem}[Welfare-optimal]\label{thm:opt_welfarw}
There is an algorithm that computes a $2n\epsilon$-optimal audit vector for \cref{eq:welfare_max} in $O(m^2)$ time for any $\epsilon>0$ and non-adaptive audit game with $(n$, $m$, $\vq$, $\val$, $\pay$, $\pen)$ and cost $\lambda$.
\end{theorem}
The algorithm is nearly identical to \Cref{alg:water_filling_algorithm}.  Since the agent’s best-response still follows from \Cref{lem:br}, we can reuse \Cref{lem:critical_approx,lem:equal_approx} and search through all critical audit vectors as in \Cref{alg:water_filling_algorithm} and return the one that maximizes the worst-case social welfare.  We omit the proof.
Similarly, we can adopt \Cref{alg:exp3} to have a no-regret algorithm for social welfare maximization.

\subsection{Monotonicity in Penalty and Audit Cost}\label{sec:non-adaptive_penality}
Besides designing the audit vector, the principal may also adjust the penalty function or face a different audit cost $\lambda$.  We show that increasing the penalty or decreasing the audit cost $\lambda$ can only improve the principal's utility and social welfare.  

Let $V_\lambda(\vp; \pen)$ and $W_\lambda(\vp; \pen)$ be the principal's worst case utility (\cref{eq:utiliy_max}) and worst-case social welfare respectively under penalty function $\pen$ and cost $\lambda$.
\begin{restatable}{proposition}{propmonotone}\label{prop:monotone}
   If $\lambda\ge \lambda'$ and $\pen(k)\le \pen'(k)$ for all $k\in [m]$, for any $\vp$ there exists $\vp'$ so that 
   $$V_\lambda(\vp; \pen) \leq V_{\lambda'}(\vp'; \pen')\text{ and }W_\lambda(\vp; \pen)\le W_{\lambda'}(\vp', \pen').$$
\end{restatable}
The idea of \cref{prop:monotone} is that if the penalty increases, we can decrease the audit probability $p_k' = \frac{\pen(k)}{\pen'(k)}p_k$, which preserves the same equilibria and expected penalty gain, but lowers the audit cost.
\section{Optimal Adaptive Audits with Costs}\label{sec:adaptive}
We now explore adaptive audit games, where the principal's strategy depends on both the agents' prior distribution $\vq$ and report distribution $\hat{\vq}$.  We defer the budgeted setting to~\cref{sec:adaptive_budgeted}, and proofs in \cref{app:adaptive}. 

In this section, we assume that the penalty is less sensitive than the payment:
\begin{equation}\label{eq:assum4}
    \frac{\pay(l)}{\pay(k)}\ge \frac{\pen(l)}{\pen(k)} \text{ for all }k\le l\in [m].
\end{equation}
For instance, any positive affine function $\pen = a\pay+b$ with $a,b\ge 0$ satisfies~\cref{eq:assum4}.

As multiple equilibria may exist, the principal optimizes for the worst-case utility by solving the following optimization problem: 
\begin{equation}\label{eq:obj_adaptive}
    \sup_{\pi: \Delta_m\to [0,1]^m}\min_{\mQ\in Eqi(\pi)}V_\lambda(\pi, \mQ).
\end{equation}  We define $V_\lambda(\pi) = \min_{\mQ\in Eqi(\pi)}V_\lambda(\pi, \mQ)$ as the principal's worst case utility and set to $-\infty$ if $Eqi(\pi) = \emptyset$ following the standard convention in pessimistic Stackelberg games~\cite{coniglio2017pessimistic}.  We say that $\pi$ \emph{$\epsilon$-approximates} $\pi'$ if 
$V_\lambda(\pi) \ge V_\lambda(\pi') -\epsilon$, and $\pi$ is \emph{$\epsilon$-optimal} if it $\epsilon$-approximates any $\pi'$.

\begin{restatable}{theorem}{thmoptadaptivecostly}\label{thm:opt_adaptive_costly}
There is an algorithm that computes an $\epsilon$-optimal audit vector for \cref{eq:obj_adaptive} with \cref{eq:util_principal} in $O(m^2)$ time for all small enough $\epsilon>0$ and adaptive audit game with cost $\lambda\ge 0$ and parameters $(n$, $m$, $\vq$, $\val$, $\pay$, $\pen)$ satisfying~\cref{eq:assum4}.
Moreover, $$\sup_{\pi}\min_{\mQ\in Eqi(\pi)}V_\lambda(\pi, \mQ) = \sup_{\vp\in [0,1]^m, \mQ\in Eqi(\vp)} V_\lambda(\vp, \mQ).$$
\end{restatable}
\paragraph{Proof sketch.}  
To prove Theorem~\ref{thm:opt_adaptive_costly}, we use three key observations. First, due to \cref{lem:br}, equilibria depend only on the output vector $\vp$. Adaptive strategies cannot yield new equilibria beyond those already attainable by some fixed $\vp$. Consequently, 
 \begin{equation}\label{eq:adaptive_up}
     V_\lambda(\pi)\le \sup_{\vp, \mQ\in Eqi(\vp)} V_\lambda(\vp, \mQ)
 \end{equation}
and we will show that \cref{eq:adaptive_up} holds with equality.  This simplifies the bi-level optimization to maximization over $(\vp, \mQ)$ pairs.

Second, we define a \emph{dictator audit strategy} with $\vp^*$ and $\hat{\vq}^*\in \Delta_m$ as
\begin{equation}\label{eq:dict}
    \pi_{dict}(\hat{\vq}) = \begin{cases}
    \vp^*&\text{ if $\hat{\vq} = \hat{\vq}^*$}\\
    \mathbf{1}&\text{ if $\hat{\vq} \neq \vq$ and $\hat{\vq}\neq \hat{\vq}^*$}\\
    \mathbf{0}&\text{ if  $\hat{\vq} = \vq$ and $\hat{\vq} \neq \hat{\vq}^*$}
\end{cases}.
\end{equation}
Intuitively, if the observed reports differ from $\vq^*$, the dictator audit strategy either audits everyone (making any misreporting strictly unprofitable) or audits no one (agents strictly prefer to over-report as the highest type).  \Cref{lem:dict} shows that a dictator audit strategy can eliminate any report strategy with $\hat{\vq}\neq \hat{\vq}^*$, while ensuring the existence of an equilibrium with $\hat{\vq} = \hat{\vq}^*$ by choosing $\vp^*$ appropriately.  
\begin{restatable}[Dictator strategies]{lemma}{lemdict}\label{lem:dict}
    For any dictator audit strategy $\pi_{dict}$ in \cref{eq:dict} with $\vp^*$ and $\hat{\vq}^*$, 
    $Eqi(\pi_{dict}) =
        \{\mQ\in Eqi(\vp^*): \hat{\vq} = \hat{\vq}^*\}.$
\end{restatable}

Finally, \Cref{lem:best_eqi} shows that for any audit vector $\vp$, the best equilibrium can be single-minded.  Therefore, it is sufficient to iterate all single-minded strategies $\mQ$ and search the optimal audit vector $\vp$ with $\mQ\in Eqi(\vp)$.  Moreover, by a similar argument as in \cref{lem:critical_approx}, we show that the optimal audit vector is critical.  This reduces the search to $O(m^2)$ candidates, yielding the claimed $O(m^2)$ running time.
\begin{restatable}{lemma}{lembesteqi}\label{lem:best_eqi}
    For any audit vector $\vp$, if \cref{eq:assum4} holds,  there exists a single-minded equilibrium $\mQ'\in Eqi(\vp)$ so that for all $\mQ\in Eqi(\vp)$
    $V_\lambda(\vp, \mQ)\le V_\lambda(\vp, \mQ')$.
\end{restatable}

\begin{remark}
    Note that the argument to prove \cref{lem:best_eqi} also applies to social welfare, so \cref{thm:opt_adaptive_costly} also holds for optimizing social welfare.  Additionally, by \cref{lem:best_eqi}, if \cref{eq:assum4} holds, \cref{alg:water_filling_algorithm} also finds an approximately optimal audit vector in the non-adaptive setting, and the worst-case utility coincides with the best-case utility $$\sup_\vp\min_{\mQ\in Eqi(\vp)}V_\lambda(\vp,\mQ) = \sup_\vp\max_{\mQ\in Eqi(\vp)}V_\lambda(\vp,\mQ).$$
    Conversely, if \cref{eq:assum4} is not satisfied, the optimal equilibrium may not be single-minded, and this equivalence no longer applies.
\end{remark}

\section{Optimal Adaptive Audits with Budget}\label{sec:adaptive_budgeted}
Instead of costly audits, the principal may have a budget constraint.  In this section, we focus on the adaptive setting, as a non-adaptive strategy can be ill-defined when the realized number of reports for a type is smaller than the number of allocated audits.  Using ideas similar to \cref{sec:adaptive}, we discuss how to design optimal audit strategies for the principal's utility and social welfare.  We also connect our work to \citet{Estornell2020StrategicClassification}, which considers another objective, minimizing misreport incentives, in~\cref{sec:minincentive}.

\subsection{Optimizing the principal's utility}
In the \emph{budgeted setting}, the principal has a fixed budget $B$ limiting the total expected number of audits. The utility function becomes\footnote{Equivalently, the principal can conduct $b_k$ audits on type $k\in [m]$ with $0\le b_k\le n\hat{q}_k$ and $\sum_k b_k\le B$, and an agent reporting as type $k$ is audited with probability $p_k = \frac{b_k}{n\hat{q}_k}$, because $n\sum_{i,k} q_iQ_{i,k}p_k = n\sum_k \hat{q}_k q_k = \sum_k b_k$.  Though either $p_k$ or $b_k$ can represent the principal's audit, we use $p_k$ to align with the costly setting.}
\begin{equation}\label{eq:util_principal_budgeted}
    V_B(\vp, \mQ) = \begin{cases}V(\vp,\mQ)&\text{ if }C(\vp,\mQ)\le B\\
        -\infty&\text{ otherwise.}
    \end{cases}
\end{equation}
As multiple equilibria may exist, the principal optimizes for the worst-case utility by solving the following
optimization problem:
\begin{equation}\label{eq:obj_adaptive_budgeted}
    \sup_{\pi: \Delta_m\to [0,1]^m}\min_{\mQ\in Eqi(\pi)}V_B(\pi, \mQ).
\end{equation}

\begin{theorem}\label{thm:opt_adaptive_budgeted}
There is an algorithm that computes an optimal audit vector for \cref{eq:obj_adaptive_budgeted} in $O(m^2)$ time for any adaptive audit game with budget $B$ and parameters $(n$, $m$, $\vq$, $\val$, $\pay$, $\pen)$ satisfying~\cref{eq:assum4}.
Moreover, $$\sup_{\pi}\min_{\mQ\in Eqi(\pi)}V_B(\pi, \mQ) = \sup_{\vp\in [0,1]^m, \mQ\in Eqi(\vp)} V_B(\vp, \mQ).$$
\end{theorem}

Let $\Pi_B(\mQ)\subseteq [0,1]^m$ be the set of audit vectors that satisfy the budget constraints,
$$\Pi_B(\mQ) = \{\vp\in [0,1]^m: C(\vp,\mQ)\le B\text{ for all }\mQ\in Eqi(\vp)\}.$$
As the budget constraint only depends on the report distribution we can also write it as $\Pi_B(\hat{\vq})$.  Moreover, the set of all budgeted audit strategies is $\Pi_B = \{\pi: Eqi(\pi)\neq \emptyset, C(\pi(\hat{\vq}), \mQ)\le B\text{ for all }\mQ\in Eqi(\pi)\}$. 
We define a \emph{dictator audit strategy} $\pi_{dict}$  with $B$, $\vp^*$ and $\hat{\vq}^*\in \Delta_m$ as the following 
\begin{equation}\label{eq:dict_budgeted}
    \pi_{dict}(\hat{\vq}) = \begin{cases}
    \vp^*&\text{ if $\hat{\vq} = \hat{\vq}^*$}\\
    \left(0,\dots,0,\min\left\{\frac{B}{n\hat{q}_{m-1}}, 1\right\}\right)&\text{ if $\hat{\vq}\neq \hat{\vq}^*$ and $\hat{q}_{m-1}> q_{m-1}$}\\
    \mathbf{0}&\text{ otherwise}
\end{cases}.
\end{equation}

Now we outline the algorithm which is similar to \cref{thm:opt_adaptive_costly}'s.  If the budget $B$ is too small everyone would (mis)report as the highest type $m-1$.  Our algorithm depends on the value of $B$.  We set 
$$\beta: = \frac{\pay(m-1)-\pay(m-2)}{\pen(m-1)}.$$

The following lemma shows that $n\beta$ is the threshold of insufficient budget so that all (mis)report as the highest type, and allocating all audit budget to $m-1$ achieves the optimality.  We defer the proofs to \cref{app:budget}

\begin{restatable}{lemma}{lemsmallbudget}\label{lem:small_budget}
    If $B\le n \beta$, all (mis)reporting as the highest type $m-1$ is an equilibrium for all $\pi\in \Pi_B$, and 
    $$V_B(\pi)\le \sum_{i} q_i(\val(i,m-1)-\pay(m-1))+\frac{B}{n}\pen(m-1).$$
\end{restatable}

For the sufficient budget case ($B>\beta n$), our algorithm is similar to the one in \cref{sec:adaptive}.   We first construct a family of dictator audit strategies for the budgeted setting in~\cref{lem:dict_budgeted}. Second, we show that it suffices to search over single-minded report strategies and a small set of equalized audit vectors (\cref{lem:best_eqi_budgeted,lem:best_equal}).


\begin{restatable}{lemma}{lemdictbudgeted}\label{lem:dict_budgeted}
    Let $m\ge 2$.  For any dictator audit strategy $\pi_{dict}$ with $B$, $\vp^*$, and $\hat{\vq}^*$ in \cref{eq:dict_budgeted}, if  $B> n\beta
    $, 
    $$Eqi(\pi_{dict}) =
        \{\mQ: \mQ\in Eqi(\vp^*)\text{ and } \hat{q} = \hat{q}^*\}$$
        and $\pi_{dict}\in \Pi_B$ if $\vp^*\in \Pi_B(\hat{\vq}^*)$.
\end{restatable}

To achieve the upper bound in \cref{thm:opt_adaptive_budgeted}, the following lemma shows that we can reduce the space of candidate report strategies to single-minded ones.  Although similar to \cref{lem:best_eqi}, \cref{lem:best_eqi_budgeted} requires a different audit vector to improve the principal’s utility and satisfy the budget constraint.

\begin{restatable}{lemma}{lembesteqibudgeted}\label{lem:best_eqi_budgeted}
    For any $\vp$ and $\mQ\in Eqi(\vp)$, there exist a single-minded $\mQ'$ with $\iota = \truthidx(\vp), \kappa = \min \hat{A}(\vp)$ and $\vp'$ where
    $$p_k' = \begin{cases}
    0& k<\iota\\
    \rho_k(\hat{u}(\vp)) &k\ge \iota
\end{cases}$$
so that $\mQ'\in Eqi(\vp')$, 
$$C(\vp',\mQ')\le C(\vp,\mQ)\text{ and }V(\vp',\mQ')\ge V(\vp,\mQ).$$
\end{restatable}

\begin{restatable}{lemma}{lembestequal}\label{lem:best_equal}
    For any single-minded $\mQ$ with $\iota, \kappa$ and $\Pi_B(\mQ)\neq\emptyset$, there exists  $u \in \mathbb{R}$ defined in~\cref{eq:best_equal0} so that $\vp'$ with
    \begin{equation}\label{eq:opt_adaptive_budgeted1}
    p_k' = \begin{cases}
    0& k<\iota\\
    \rho_k(u) &k\ge \iota
\end{cases}
\end{equation} satisfies
$V(\vp',\mQ)\ge V(\vp,\mQ)$ for all $\vp \in \Pi_B(\mQ)$.
\end{restatable}

\begin{proof}[Proof of \cref{thm:opt_adaptive_budgeted}]

If $B\le \beta n$, we allocate all audit budget to $m-1$ which achieves the upper bound in \cref{lem:small_budget}.  Otherwise, by \cref{lem:best_eqi_budgeted}, we can go through all single-minded report strategies and compute, via Lemma \ref{lem:best_equal}, an optimal audit vector $\vp'\in \Pi_B(\mQ)$, such that the optimal pair $\vp^*, \mQ^*$ satisfy
$$V_B(\vp^*,\mQ^*) = \sup_{\mQ\in Eqi}\sup_{\vp\in \Pi_B(\mQ)}V_B(\vp,\mQ).$$
because $V_B(\vp,\mQ) = V(\vp,\mQ)$ if $\vp\in \Pi_B(\mQ)$.
We return a dictator strategy $\pi^*$ with $B, \vp^*$ and the report distribution of $\mQ^*$.  By \cref{lem:dict_budgeted,lem:br}, 
$$V_B(\pi^*) = \min_{\mQ: \mQ\in Eqi(\vp^*), \hat{q} = \hat{q}^*}V_B(\vp^*,\mQ) = V_B(\vp^*,\mQ^*)$$ 
because a given report distribution supports a unique equilibrium. Therefore $\pi^*$ achieves the upper bound in \cref{thm:opt_adaptive_budgeted}.

Finally, for time complexity, there are $O(m^2)$ single-minded report strategies, and we can use dynamic programming to solve all $u$ in~\cref{lem:best_equal} in $O(m^2)$.  We then use $O(m^2)$ to find the optimal single-minded strategy as \cref{alg:water_filling_algorithm}.
\end{proof}
\subsection{Minimizing the Misreport Incentive}\label{sec:minincentive}
We study the misreport incentive objective in~\citet{Estornell2020StrategicClassification} in our model, and discuss the connection to theirs.

The \emph{misreport incentive} of an audit strategy $\pi: \Delta_m\to [0,1]^m$ is 
\begin{equation*}
  \epsilon_{MI}(\pi):=\max_{i,k}\hat{U}_k(\pi(\vq))-{U}_i  
\end{equation*}
which is an agent's largest gain from misreporting.  

Since $\pi$ only affects misreport utility $\hat{U}_k(\vp) = U_k-p_k\pen(k)$, minimizing $\epsilon_{MI}$ is equivalent to finding $\vp = \pi_k(\hat{\vq})$ for each $\hat{\vq}$ such that 
\begin{equation}\label{eq:mi_nonatomic}
    \begin{aligned}
    \min_\vp \max_k& (U_k-\pen(k)p_k)\\
    &\text{ subject to } n\sum_k p_k \hat{q}_k\le B
\end{aligned}
\end{equation}
We can use linear programming to solve the optimization problem in~\cref{eq:mi_nonatomic} for a given report $\hat{\vq}$ and get the following theorem.\footnote{
Note that the algorithm implements $\pi$ as an efficient oracle: given any report $\hat{\vq}$, it computes $\pi(\hat{\vq})$ in polynomial time.  However, we cannot enumerate all outputs of $\pi$, as the domain of reports is exponential.}  Additionally, we may use the classical water-filling algorithm which sorts $k$ according to $U_k$ and allocate budget to minimize the largest misreport utility according to budget constraint.
\begin{theorem}
There is an efficient algorithm that computes the optimal audit strategy $\pi$ to minimize the misreport incentive~\cref{eq:mi_nonatomic} for all $(n,m,\vq,\val,\pay,\pen)$ and $B$.
\end{theorem}

Finally, under the strategic classification setting of \citet{Estornell2020StrategicClassification}--- 
every type is either above the acceptance threshold (payment $1$, penalty $1+b$) or below it (payment $0$, no penalty)--- the optimal solution to \cref{eq:mi_nonatomic} assigns the same auditing probability to all types above the threshold, so they are audited uniformly. This coincides with \citet[Theorem 2]{Estornell2020StrategicClassification}, which, however, is under the atomic setting.

\section{Simulations}
Thus far, we have analyzed the optimal audit policy theoretically. We now provide simulations to illustrate how the optimal policy and the resulting equilibria depend on key model parameters: the prior distribution (\cref{sec:exprior}), the audit cost, the penalty margin (\cref{sec:excost}), and the resolution of the type space (\cref{sec:extypes}).

\subsection{Effect of the prior}\label{sec:exprior}
We begin with small three-type examples that vary the prior.  \Cref{fig:heat_map} illustrates the effect of the prior $\vq$.  In the lower-left corner, most agents have the lowest type (type $0$), which admits the truthful equilibrium $(0,1,2)$.  At the top corner, most agents have type $2$, and it becomes preferable to allow everyone to report the highest type $(2,2,2)$ rather than impose huge audit costs to enforce truth-telling.   Similarly, in the lower-right corner, it is optimal to allow type $0$ to misreport as type $1$.  Finally, we note that the principal-optimal policy in  \cref{fig:heat_map1} is stricter than the welfare-optimal one in \cref{fig:heat_map2}, and yields a larger truth-telling region.   This is because misreports impose greater costs on the principal than on overall welfare.



\begin{figure}[ht]
    \centering
    \begin{subfigure}{0.5\columnwidth}
        \includegraphics[width=0.8\columnwidth]{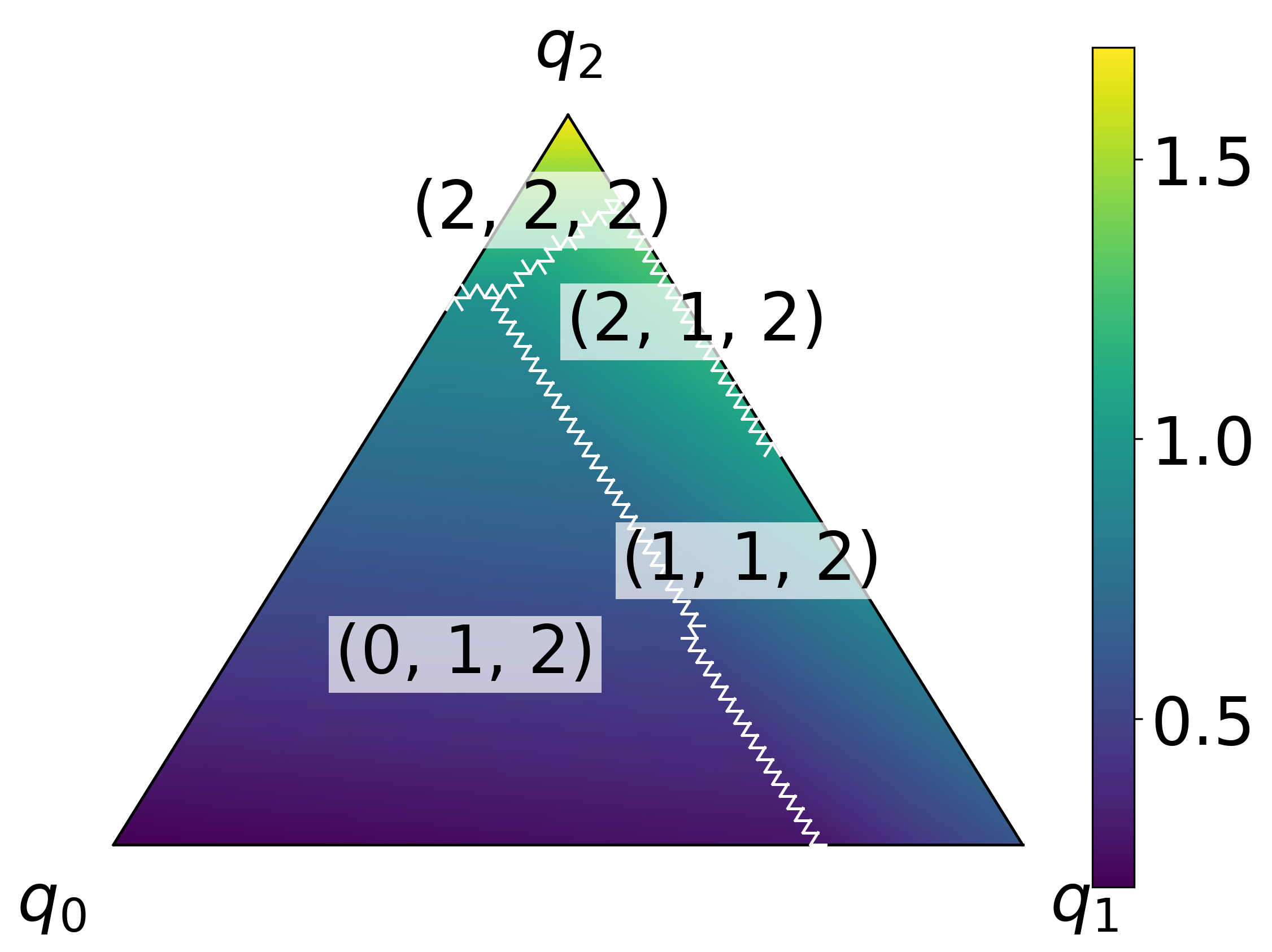}
        \caption{Principal's utility}\label{fig:heat_map1}
    \end{subfigure}~
    \begin{subfigure}{0.5\columnwidth}
        \includegraphics[width=0.8\columnwidth]{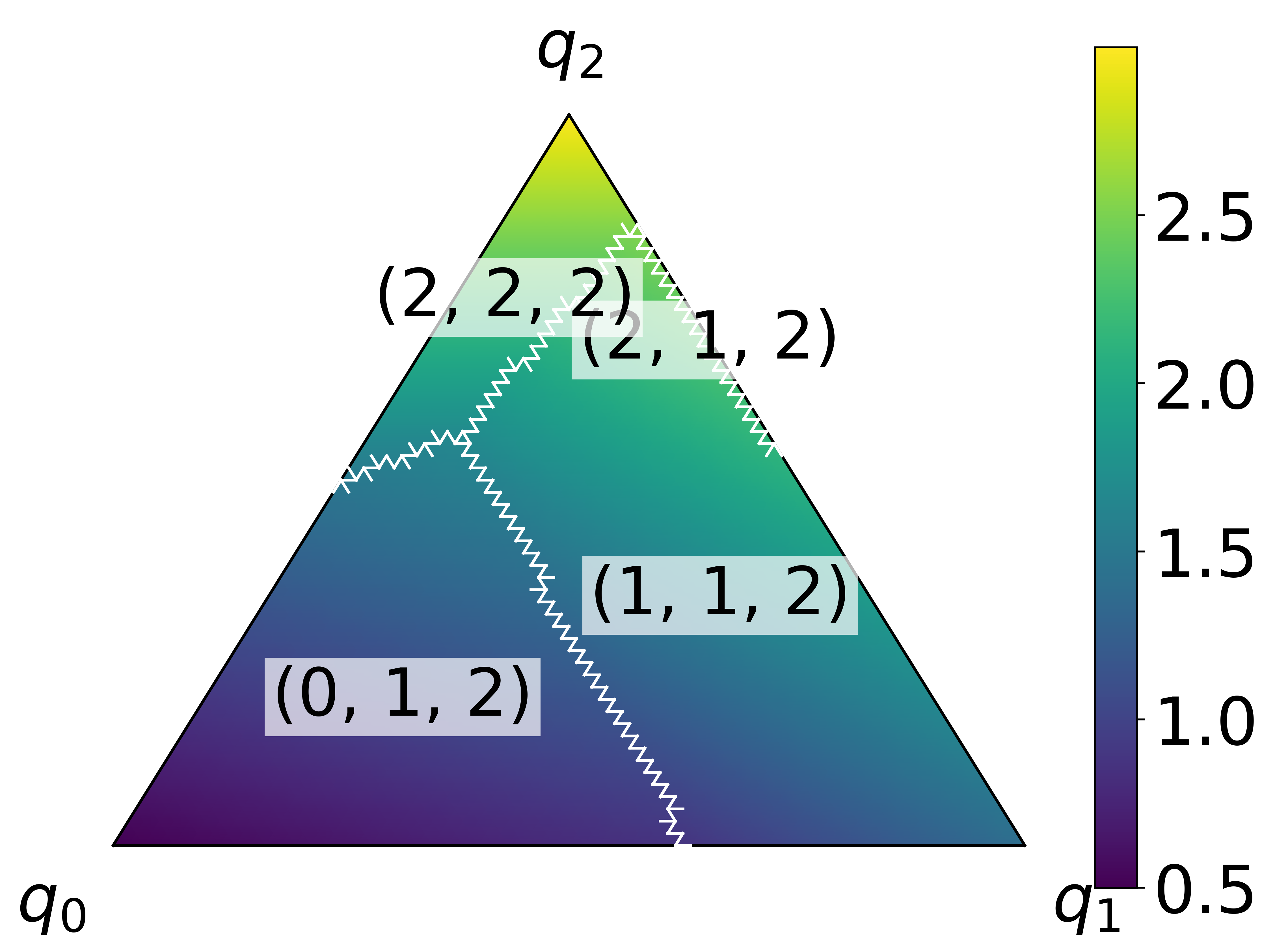}
        \caption{Social Welfare}\label{fig:heat_map2}
    \end{subfigure}
    \caption{Effect of prior $\vq$: There are three types $m = 3$ with $n = 1$, $\val = \mathrm{diag}(0.5,1.4,3.0)$, $\pay = (0.3,0.8,1.3)$, $\pen =(1.0,1.2,1.4)$, and $\lambda = 0.7$.   Each point corresponds to a prior vector $\vq = (q_0, q_1, q_2)$, and the color encodes the principal’s optimal utility by \cref{thm:opt_vector} with $\epsilon = 10^{-3}$ in \cref{fig:heat_map1}, and the optimal social welfare by  \cref{thm:opt_welfarw} in \cref{fig:heat_map2}.  We also indicate the region of the worst equilibrium.}\label{fig:heat_map}
\end{figure}

\Cref{fig:pay} shows that the effect of the payment function is non-monotone when all other parameters are fixed.  In \cref{fig:pay1}, the worst equilibrium is always truth-telling, and increasing the payments monotonically decreases the principal’s utility. In contrast, in \cref{fig:pay2}, when the type-1 payment is small, the equilibrium is still truth-telling and welfare decreases. However, for large type-1 payment, type $0$ agents begin to misreport as type $1$, and increasing $\pay(1)$ reduces audit probability $p_2$ and increases welfare.

\subsection{Impact of audit cost and penalty margin}\label{sec:excost}
In this section, we study the impact of audit cost $\lambda$ and penalty margin $b$ for affine penalty in~\cref{fig:m3_cost_sweep,fig:m3_penalty_sweep}.  
There are $m=3$ ordered types, total agent mass $n=1$, prior $\vq=(q_0,q_1,q_2)=(0.6488,\,0.3333,\,0.0179)$, the principal’s valuation matrix
\[
\val=\begin{pmatrix}
2.2 & 0.7 & 0.0\\
1.9 & 3.4 & 1.9\\
1.6 & 1.1 & 4.6
\end{pmatrix}.
\]
where the rows are the true type and columns are reported types, and the payment vector is $\pay=(1,2,3)$.
We consider affine penalties of the form $\pen=\pay+b$ where $\pen(k)=\pay(k)+b$
and costly audit with~$\lambda$. For each sweep value $\lambda$ or $b$, we compute an $\epsilon$-optimal non-adaptive audit vector for the principal's utility (U-opt) using \cref{alg:water_filling_algorithm} or social welfare (W-opt) as in~\cref{thm:opt_welfarw} with $\epsilon=10^{-3}$, and evaluate outcomes at the worst Bayes--Nash equilibrium as~\cref{eq:utiliy_max}.    (the same in the rest of the subsections)

\Cref{fig:m3_cost_sweep} varies the audit cost $\lambda$ from $0.6$ to $0.9$ with affine penalties $\pen=\pay+1.5$, and plots the principal's utility and social welfare in \cref{fig:m3_U_lam,fig:m3_W_lam}, respectively. Over this range, the optimal audit profile stays in the same regime under both objectives (utility and welfare) and is visually unchanged (\cref{fig:m3_pU_lam,fig:m3_pW_lam}): the lowest report (blue) is never audited, while the audit probabilities on the two higher reports (green and orange) remain constant. As a result, both the principal's utility and social welfare decrease with $\lambda$ because audits become more expensive, consistent with \cref{prop:monotone}.

\Cref{fig:m3_penalty_sweep} varies the penalty margin $b$ while fixing $\lambda=0.7$. \Cref{fig:m3_U_b,fig:m3_W_b} show that increasing $b$ improves both the principal's utility and social welfare, again in line with \cref{prop:monotone}. At the policy level, larger penalties substitute for auditing under both objectives: the audit probabilities on the two higher reports (green and orange) decrease as $b$ grows (\cref{fig:m3_pU_b,fig:m3_pW_b}), while auditing the lowest report (blue) remains negligible.

\begin{figure*}[t]
    \centering
    \begin{subfigure}[t]{0.4\textwidth}
        \centering
        \includegraphics[width=\linewidth]{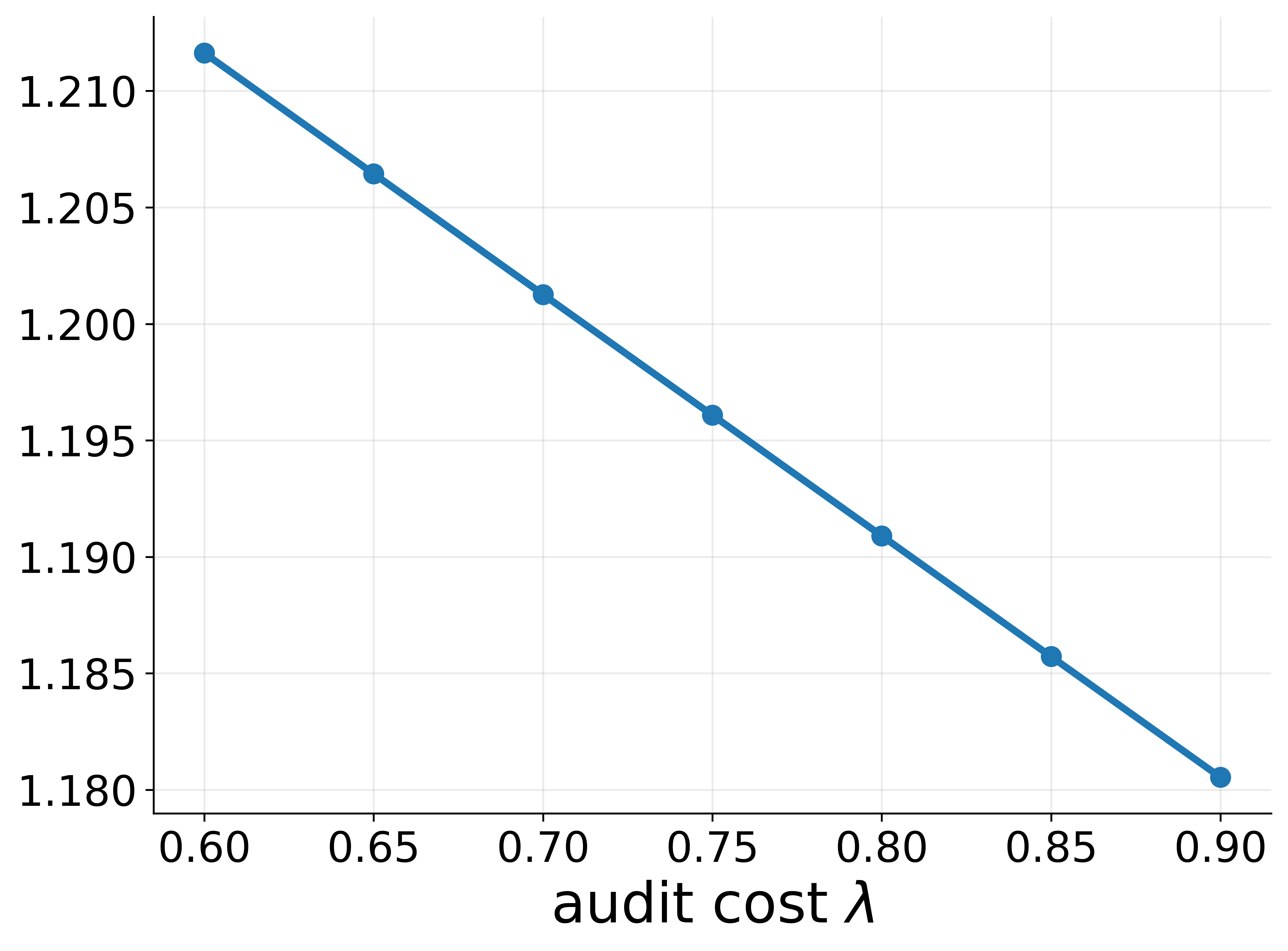}
        \caption{U-opt Principal's utility}
        \label{fig:m3_U_lam}
    \end{subfigure}\hfill
    \begin{subfigure}[t]{0.4\textwidth}
        \centering
        \includegraphics[width=\linewidth]{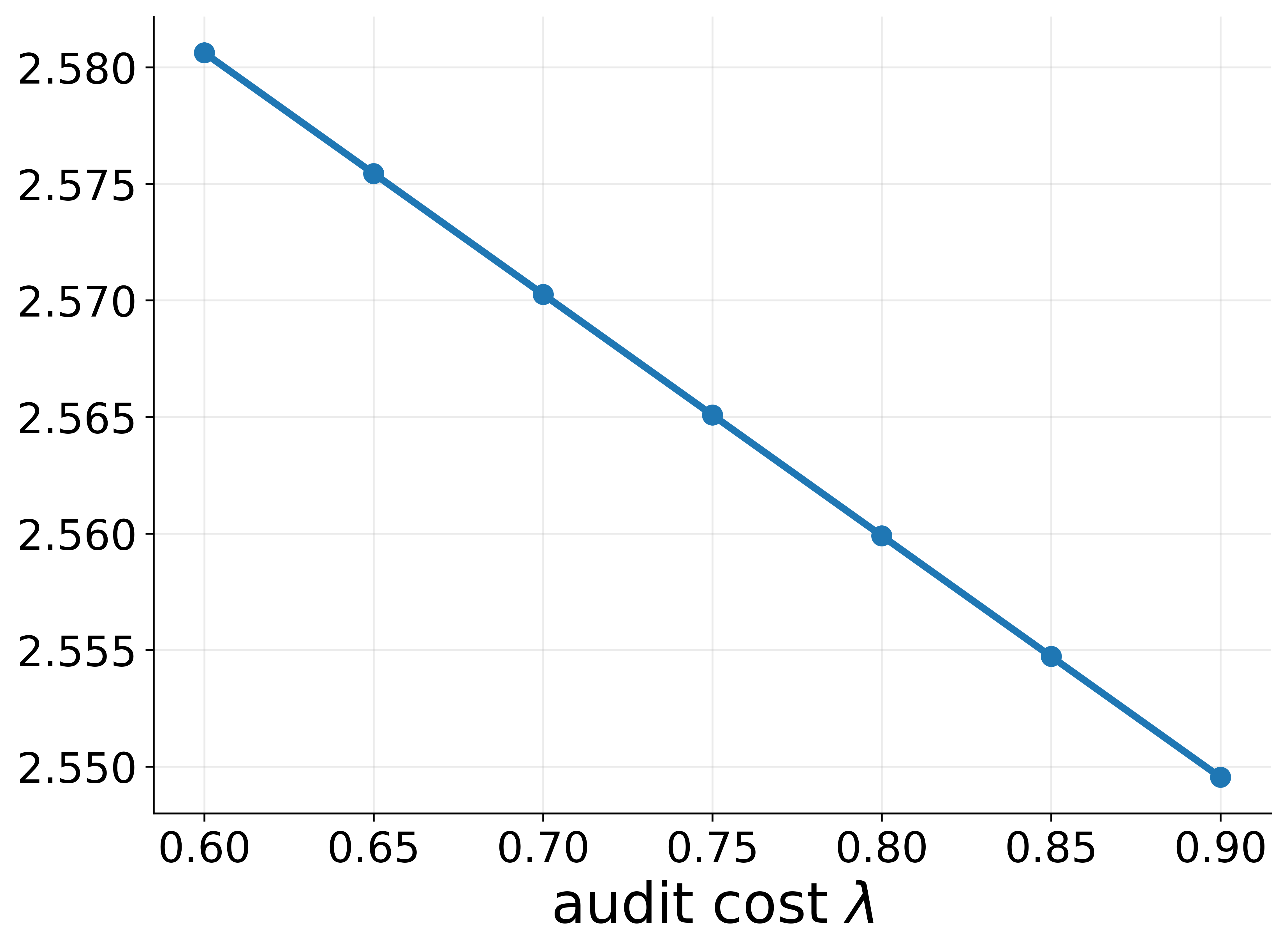}
        \caption{W-opt Social welfare}
        \label{fig:m3_W_lam}
    \end{subfigure}\hfill
    \begin{subfigure}[t]{0.4\textwidth}
        \centering
        \includegraphics[width=\linewidth]{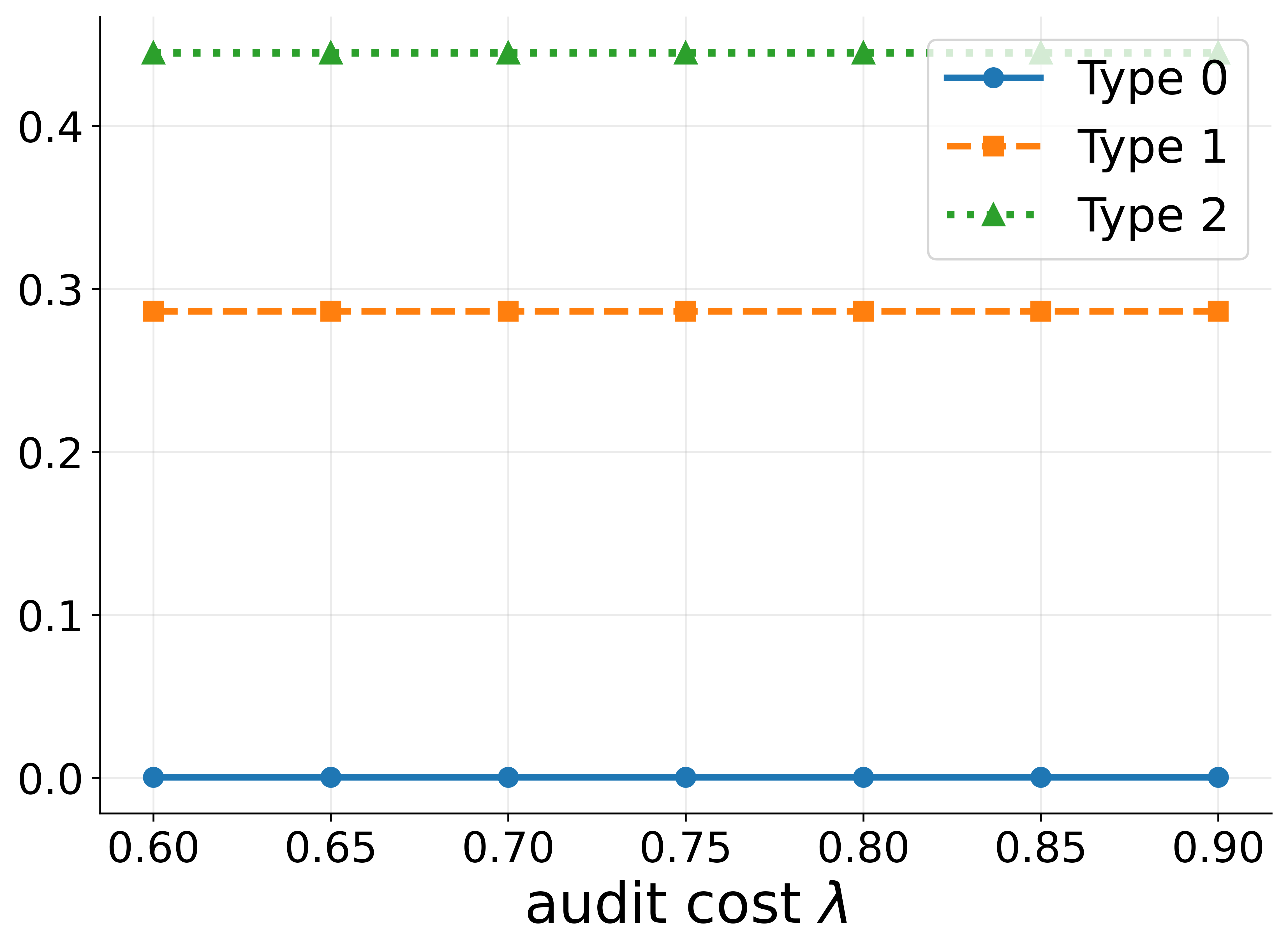}
        \caption{U-opt audit probabilities}
        \label{fig:m3_pU_lam}
    \end{subfigure}\hfill
    \begin{subfigure}[t]{0.4\textwidth}
        \centering
        \includegraphics[width=\linewidth]{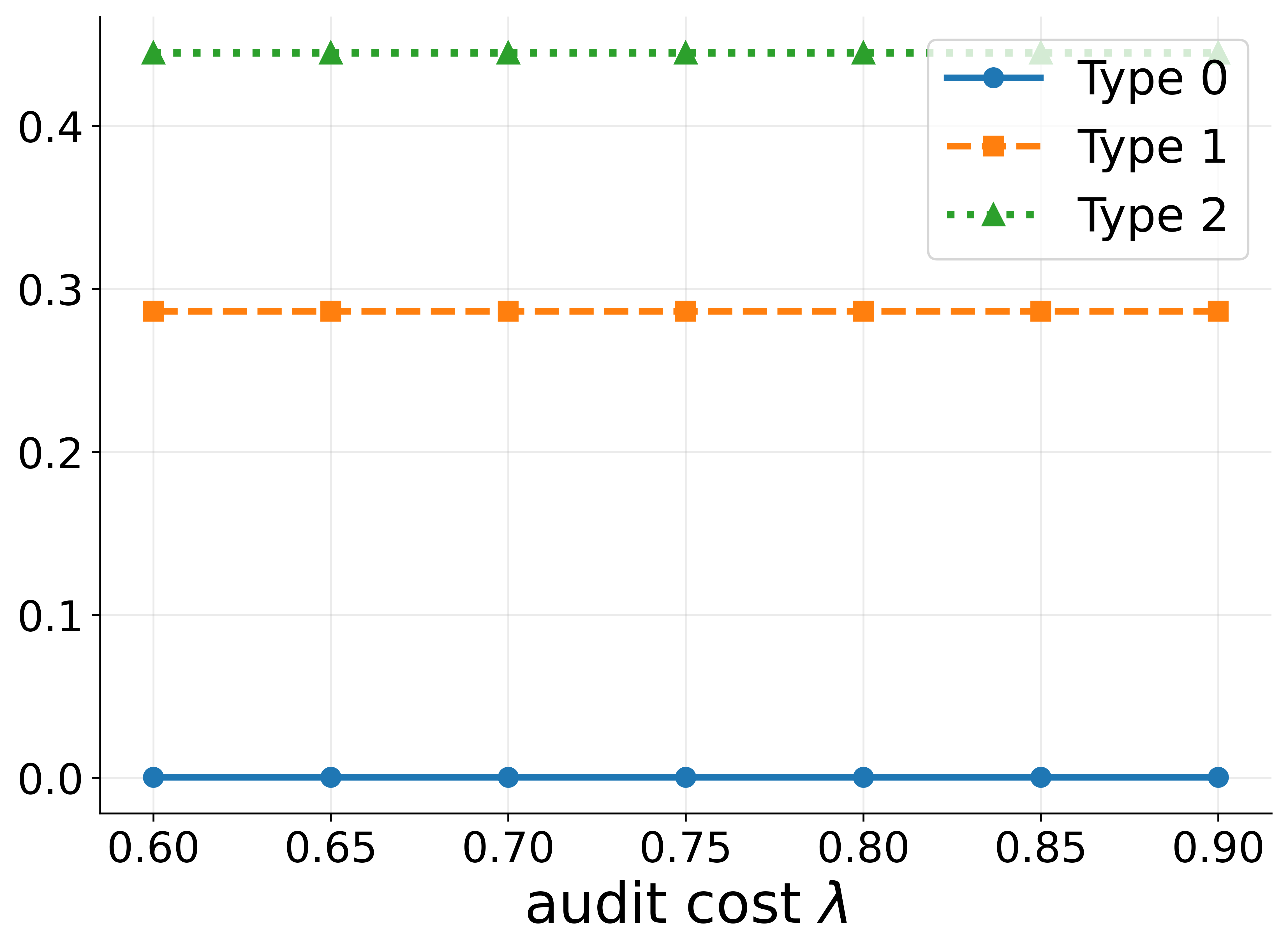}
        \caption{W-opt audit probabilities }
        \label{fig:m3_pW_lam}
    \end{subfigure}

    \caption{Effect of the per-audit cost $\lambda$ with $\pen=\pay+1.5$.}
    \label{fig:m3_cost_sweep}
\end{figure*}

\begin{figure*}[t]
    \centering
    \begin{subfigure}[t]{0.4\textwidth}
        \centering
        \includegraphics[width=\linewidth]{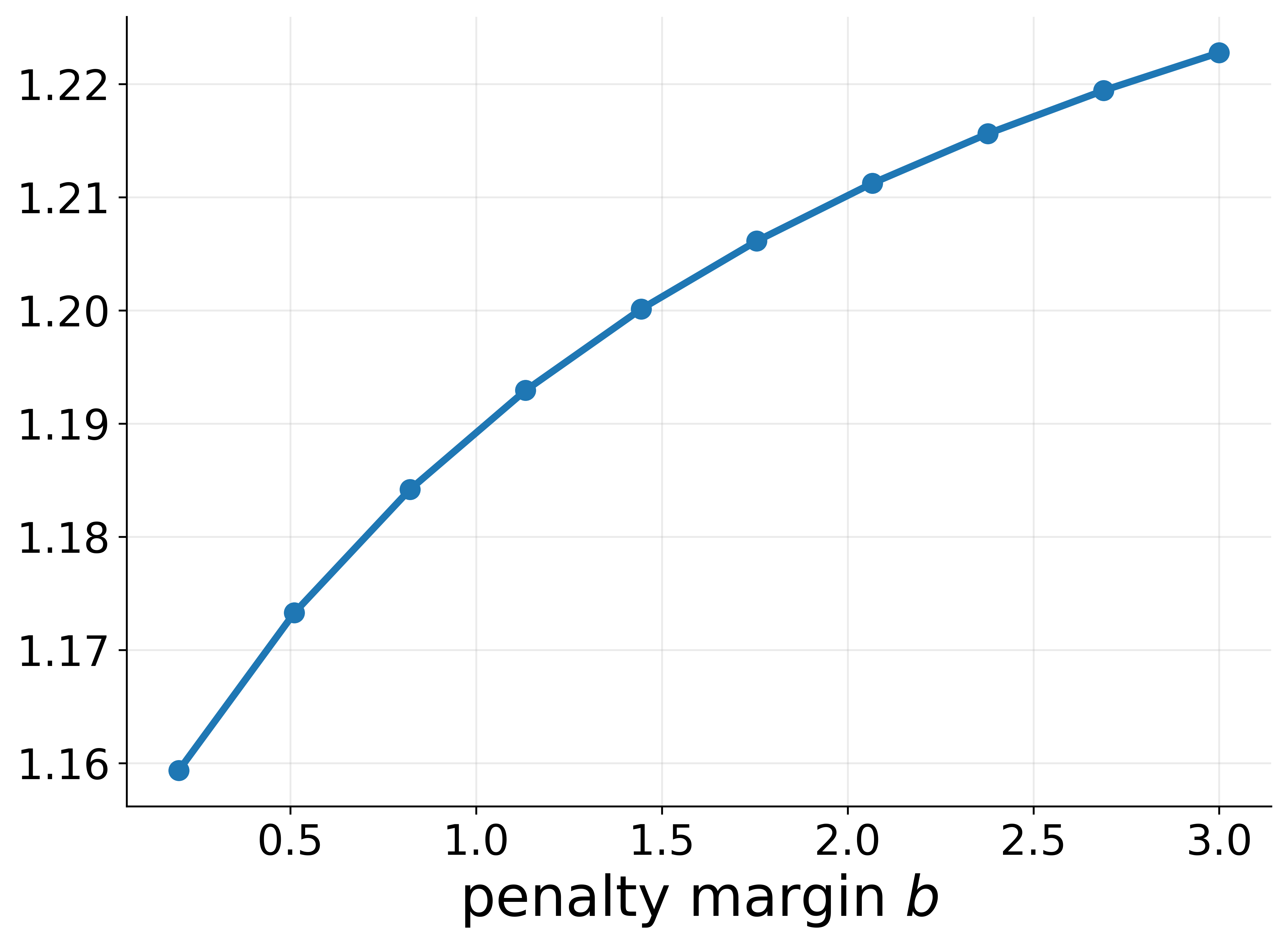}
        \caption{U-opt Principal's utility}
        \label{fig:m3_U_b}
    \end{subfigure}\hfill
    \begin{subfigure}[t]{0.4\textwidth}
        \centering
        \includegraphics[width=\linewidth]{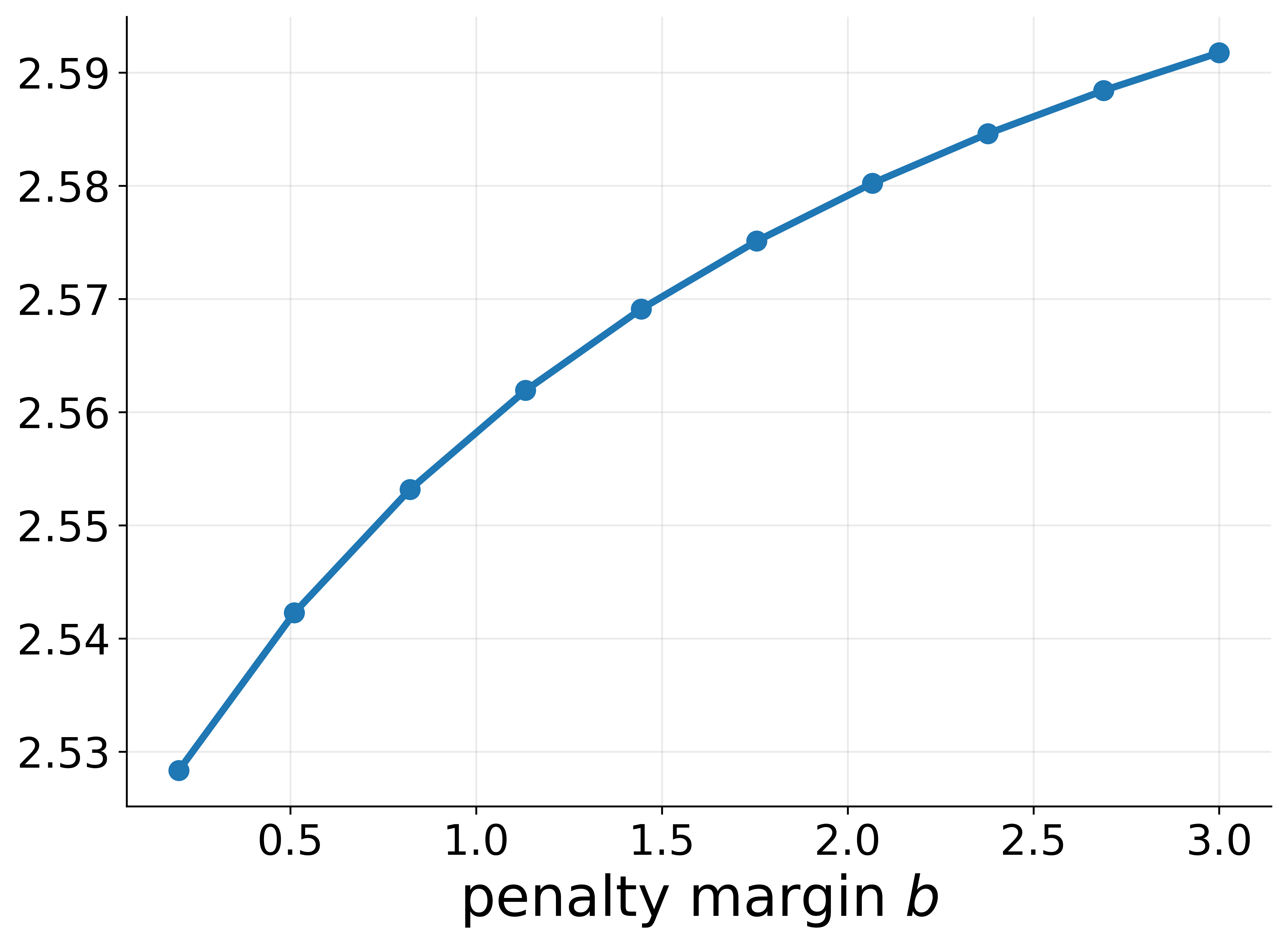}
        \caption{W-opt Social welfare}
        \label{fig:m3_W_b}
    \end{subfigure}\hfill
    \begin{subfigure}[t]{0.4\textwidth}
        \centering
        \includegraphics[width=\linewidth]{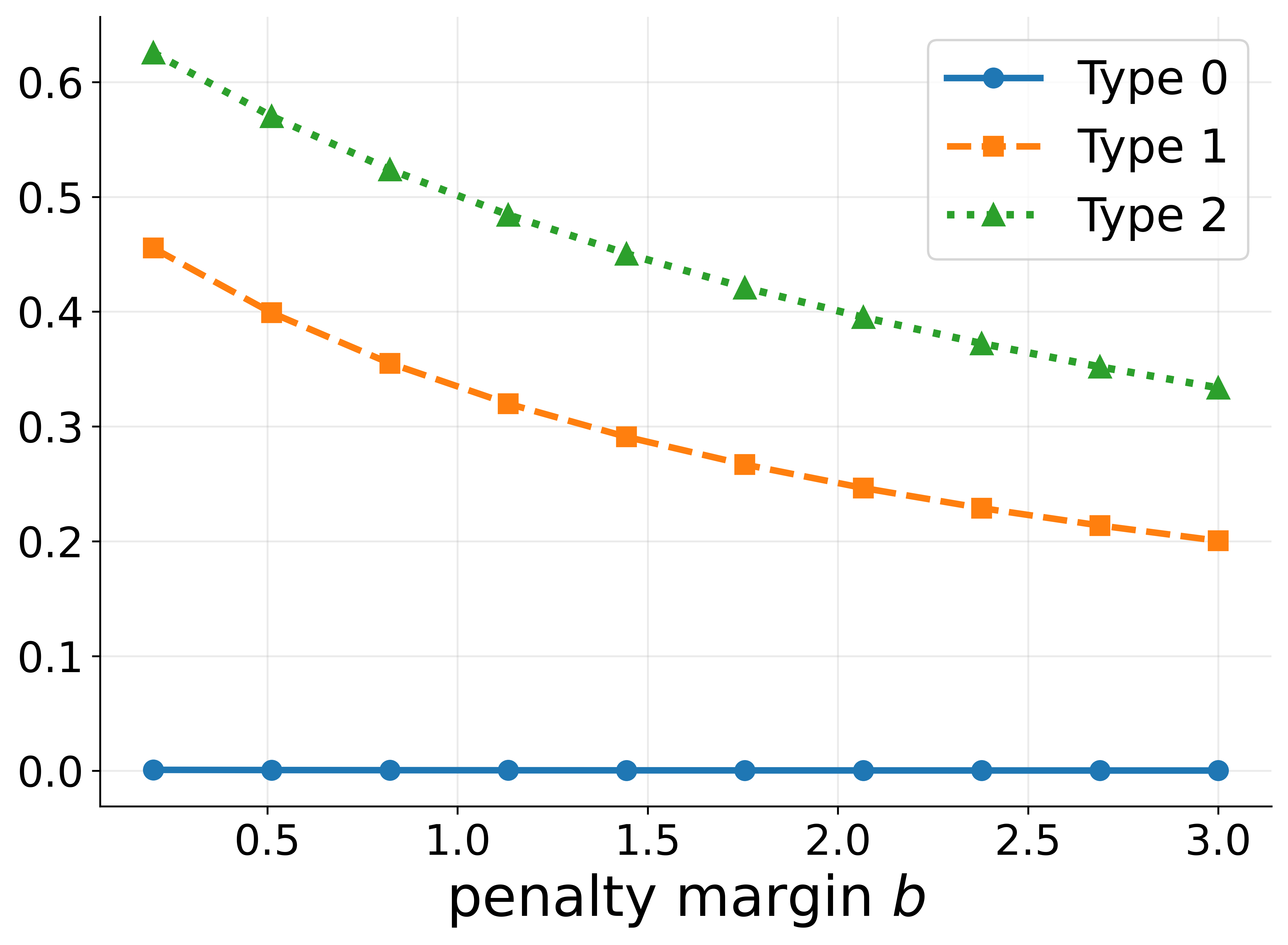}
        \caption{U-opt audit probabilities}
        \label{fig:m3_pU_b}
    \end{subfigure}\hfill
    \begin{subfigure}[t]{0.4\textwidth}
        \centering
        \includegraphics[width=\linewidth]{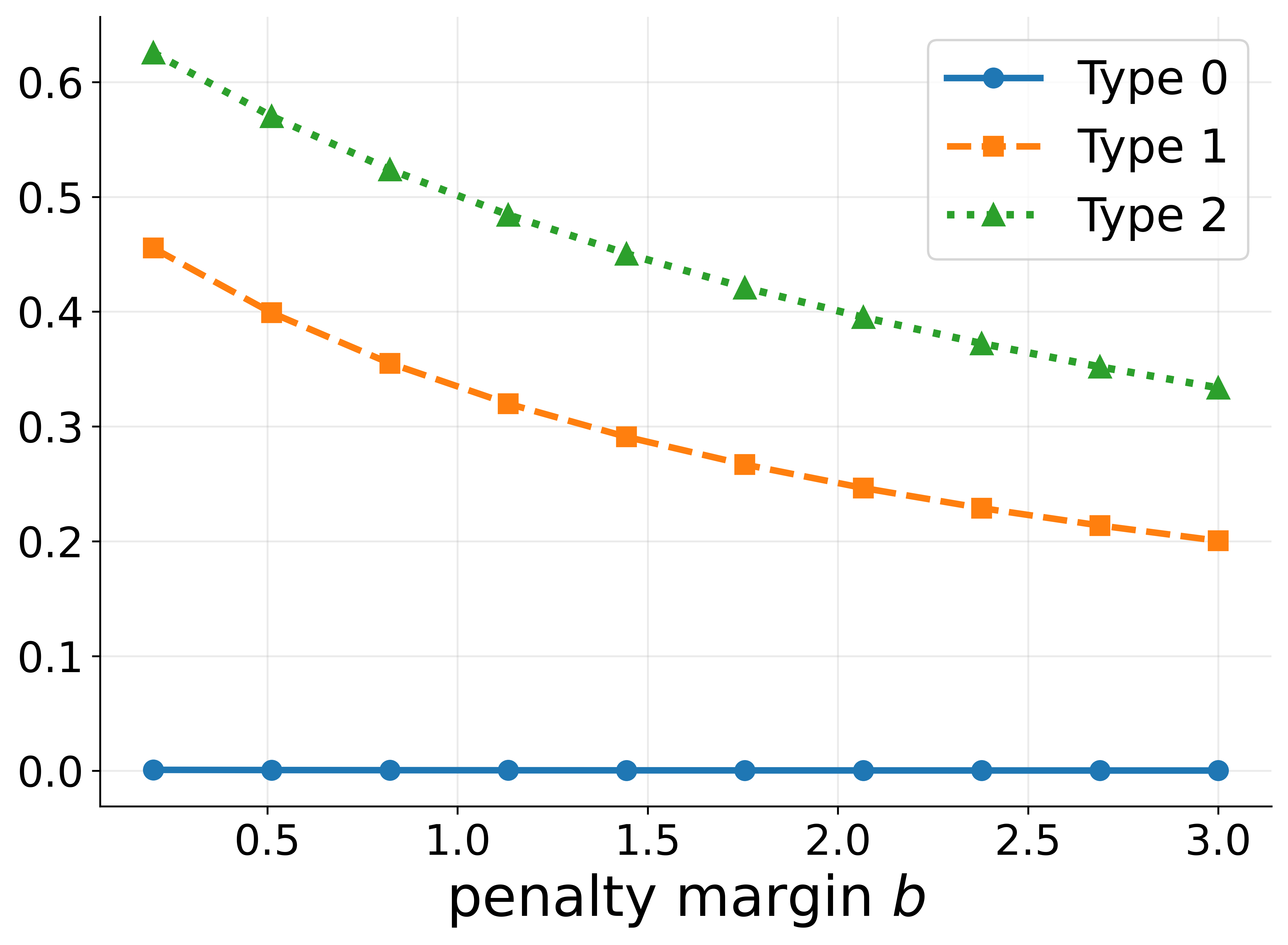}
        \caption{W-opt audit probabilities}
        \label{fig:m3_pW_b}
    \end{subfigure}

    \caption{Effect of the penalty margin $b$ with $\pen=\pay+b$ and $\lambda=0.7$.}
    \label{fig:m3_penalty_sweep}
\end{figure*}



\subsection{Resolution effect: sweeping the number of types \texorpdfstring{$m$}{m}}\label{sec:extypes}
In this section, we study how the number of types $m$ affects the principal's utility and social welfare. Intuitively, we consider a continuous population indexed by a real-valued position $x\in[0,1]$ and ask what happens when this continuous space is partitioned into $m$ ordered buckets. This captures settings in which an underlying continuous attribute---for example, annual income for a tax return---is represented by a finite set of categories. A larger $m$ gives a finer discretization of the same environment. Our goal is to understand whether this additional resolution improves objectives or primarily refines the new induced equilibrium under an optimal non-adaptive audit policy.

Specifically, we fix a continuous environment and vary only its discretization. We consider a unit-mass population ($n=1$) with true position $x$ drawn from the uniform prior on $[0,1]$. For all $x$, payments are affine,
\[
\pay(x)=1+2x,
\]
the penalty for a detected misreport is
\[
\pen(x)=\pay(x)+2,
\]
and the principal's valuation is
\[
\val(x,y)=2+2x-1|x-y|.
\]
The per audit cost is $\lambda=2.5$.

For each resolution $m$, we partition $[0,1]$ into $m$ equal-width intervals
\[
I_i=\Big[\frac{i}{m},\frac{i+1}{m}\Big),\qquad i\in\{0,1,\dots,m-1\},
\]
and use these bins as the discrete types. We construct the resulting $m$-type instance by taking bin expectations. As the prior on $x$ is uniform, each bin has mass
\[
q_i=\frac{1}{m}.
\]
Let
\[
x_i:=\mathbb{E}[X\mid X\in I_i]=\frac{2i+1}{2m}
\]
denote the mean position of bin $i$. Then the discrete payment is
\[
\pay(i)=\mathbb{E}[\pay(X)\mid X\in I_i]
      =1+2x_i
      =1+\frac{2i+1}{m},
\]
and the discrete penalty is
\[
\pen(i)=\pay(i)+2.
\]

The discrete valuation matrix is defined similarly:
\[
\val(i,k)
=
\mathbb{E}[\val(X,Y)\mid X\in I_i,\;Y\in I_k]
=
\begin{cases}
2+2x_i-\frac{|i-k|}{m}, & i\neq k,\\[6pt]
2+2x_i-\frac{1}{3m}, & i=k.
\end{cases}
\]

This construction keeps the underlying continuous model fixed and changes only the resolution m of the discretization.

For each $m\in\{2,3,\dots,200\}$, we compute an $\epsilon$-optimal non-adaptive audit vector for the principal's utility objective (U-opt) and for the social welfare objective (W-opt), with $\epsilon=10^{-6}$. As throughout the paper, outcomes are evaluated at the worst Bayes--Nash equilibrium. 




\begin{figure*}[ht]
    \centering

    \begin{subfigure}[t]{0.38\textwidth}
        \centering
        \includegraphics[width=\linewidth]{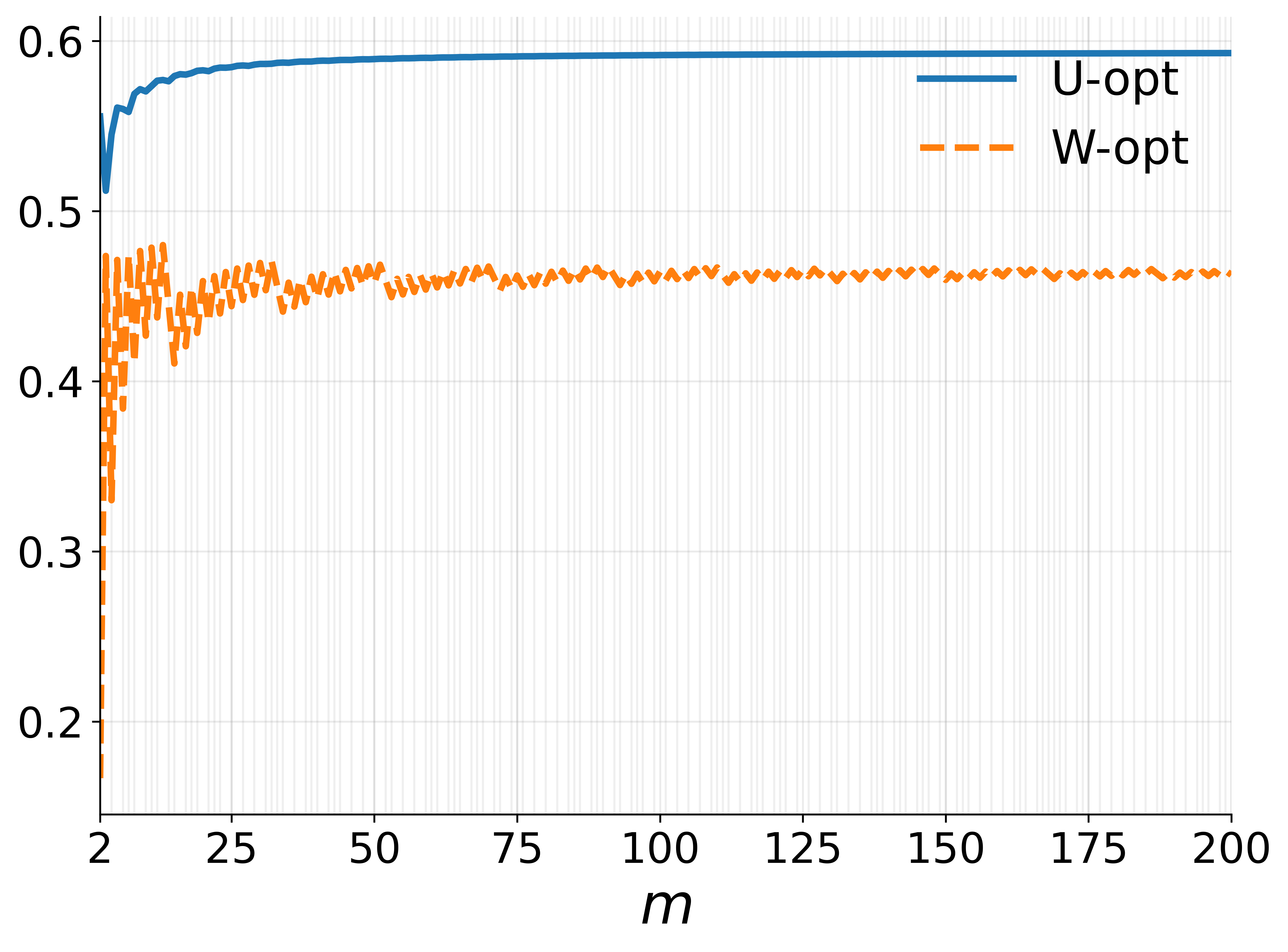}
        \caption{Principal utility.}
        \label{fig:numtypes_U}
    \end{subfigure}\hfill
    \begin{subfigure}[t]{0.38\textwidth}
        \centering
        \includegraphics[width=\linewidth]{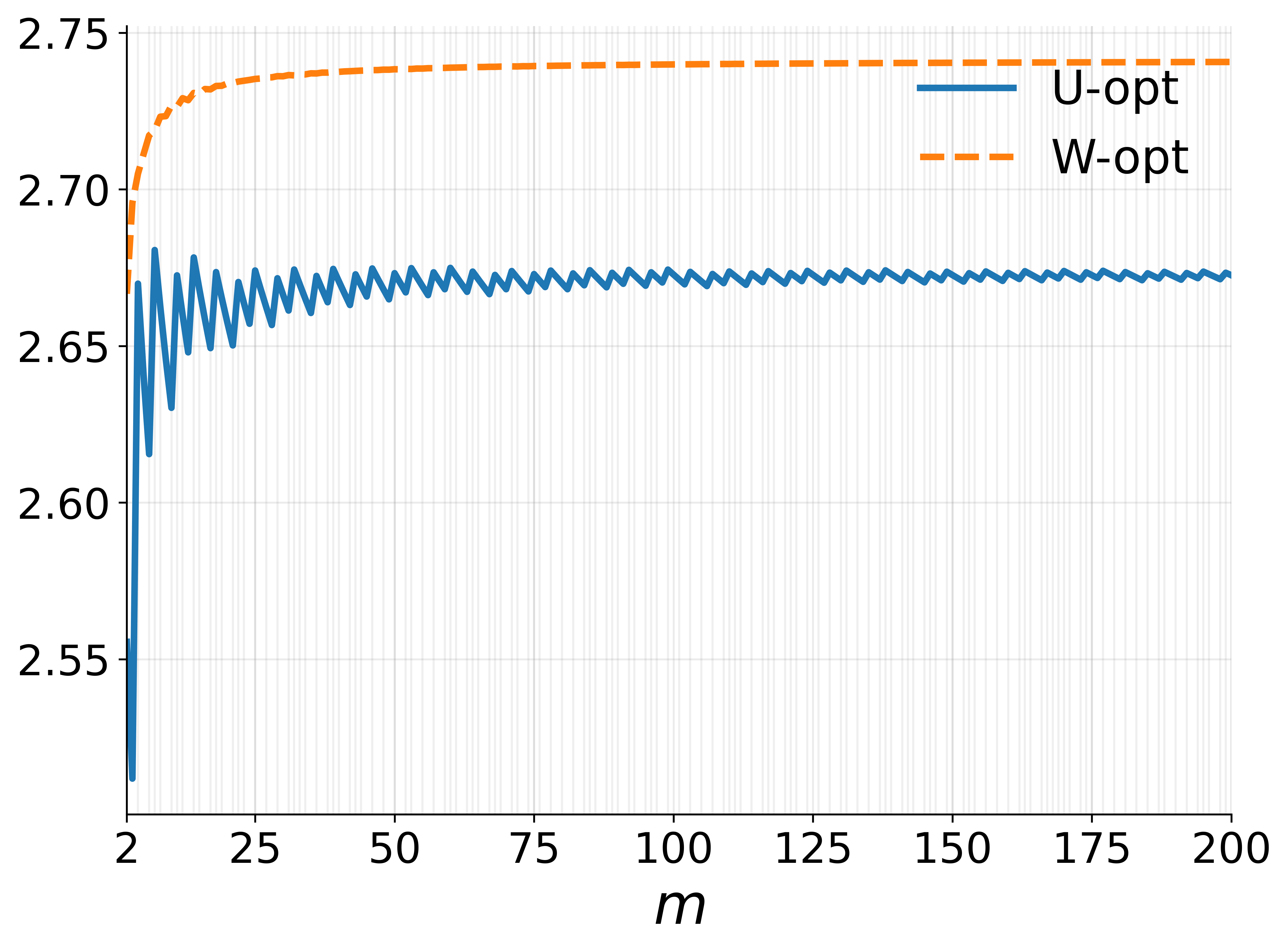}
        \caption{Social welfare.}
        \label{fig:numtypes_W}
    \end{subfigure}\hfill
    \begin{subfigure}[t]{0.38\textwidth}
        \centering
        \includegraphics[width=\linewidth]{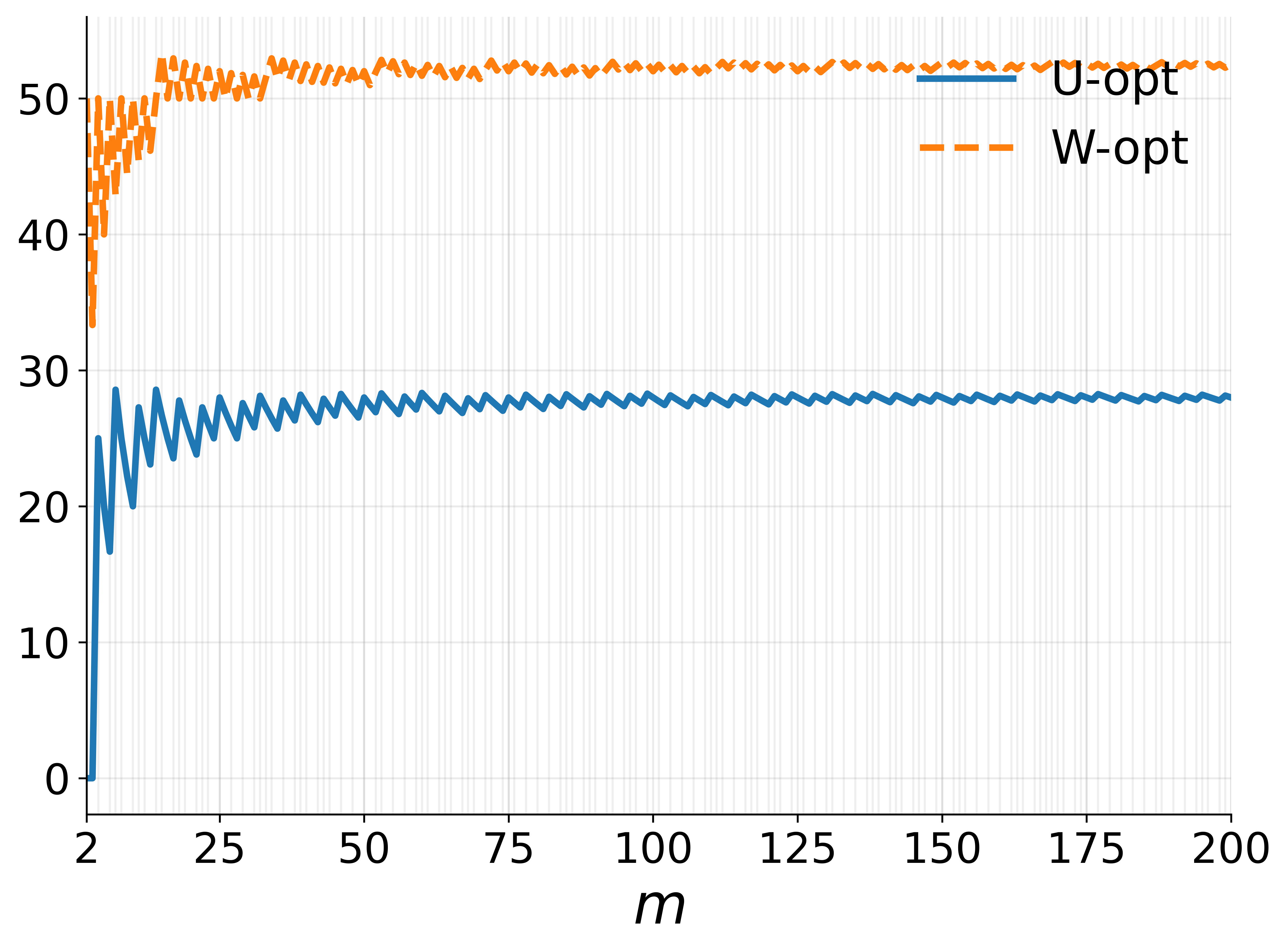}
        \caption{Misreporting mass.}
        \label{fig:numtypes_mis}
    \end{subfigure}\hfill
    \begin{subfigure}[t]{0.38\textwidth}
        \centering
        \includegraphics[width=\linewidth]{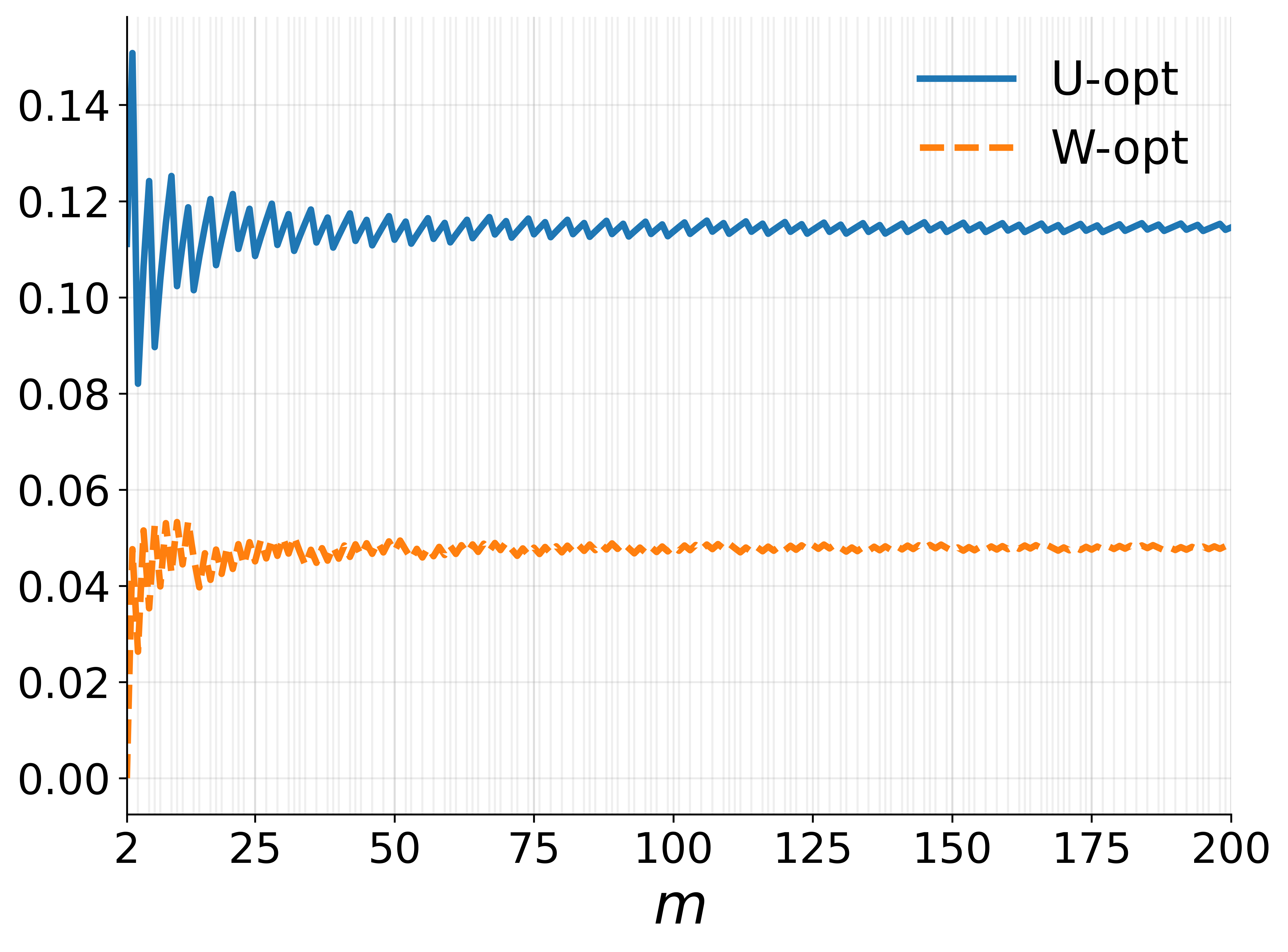}
        \caption{Audit rate.}
        \label{fig:numtypes_audit}
    \end{subfigure}

    \vspace{0.35em}

    \begin{subfigure}[t]{0.38\textwidth}
        \centering
        \includegraphics[width=\linewidth]{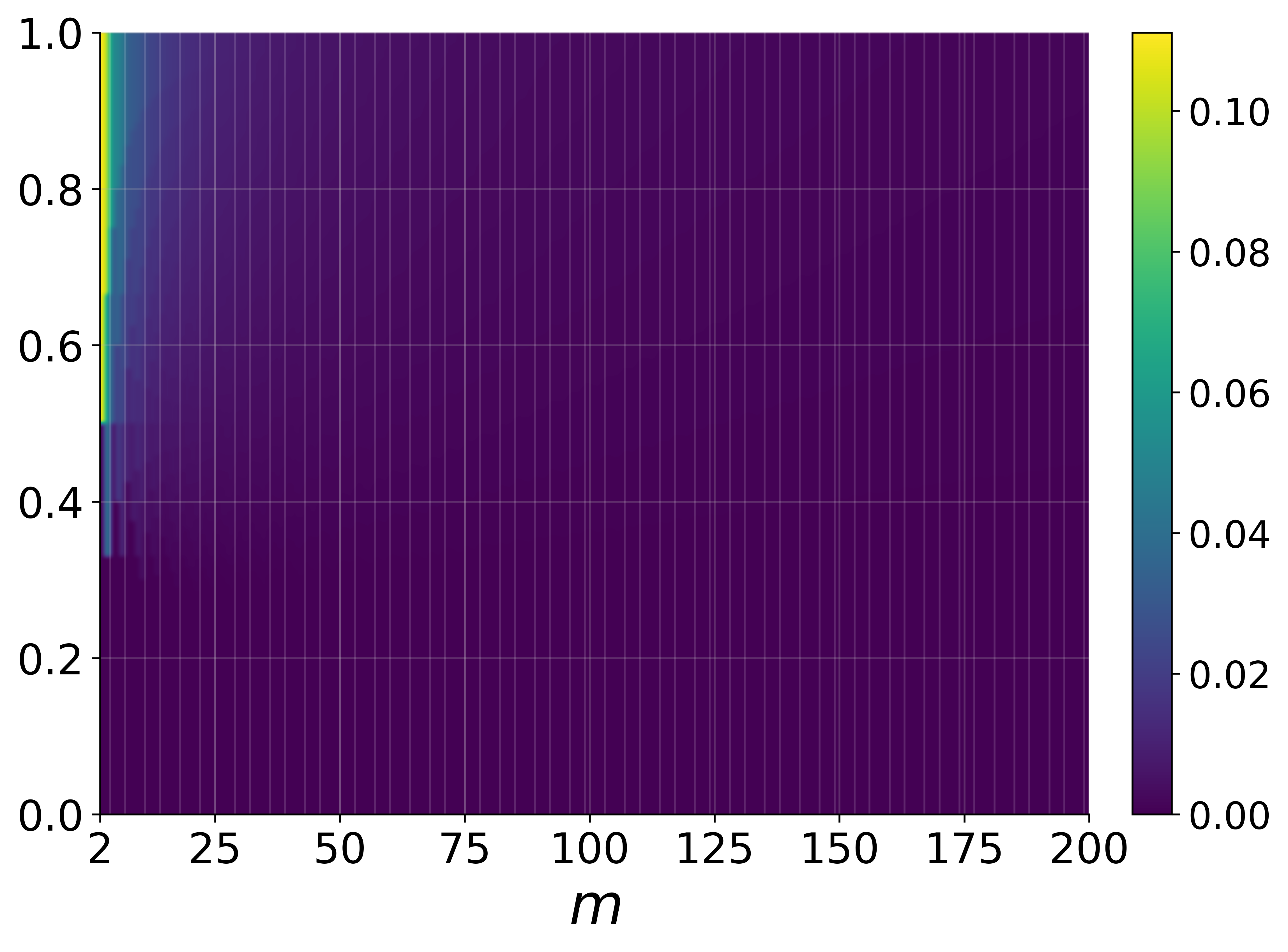}
        \caption{U-opt Audit probabilities.}
        \label{fig:numtypes_auditmass_Uopt}
    \end{subfigure}\hfill
    \begin{subfigure}[t]{0.38\textwidth}
        \centering
        \includegraphics[width=\linewidth]{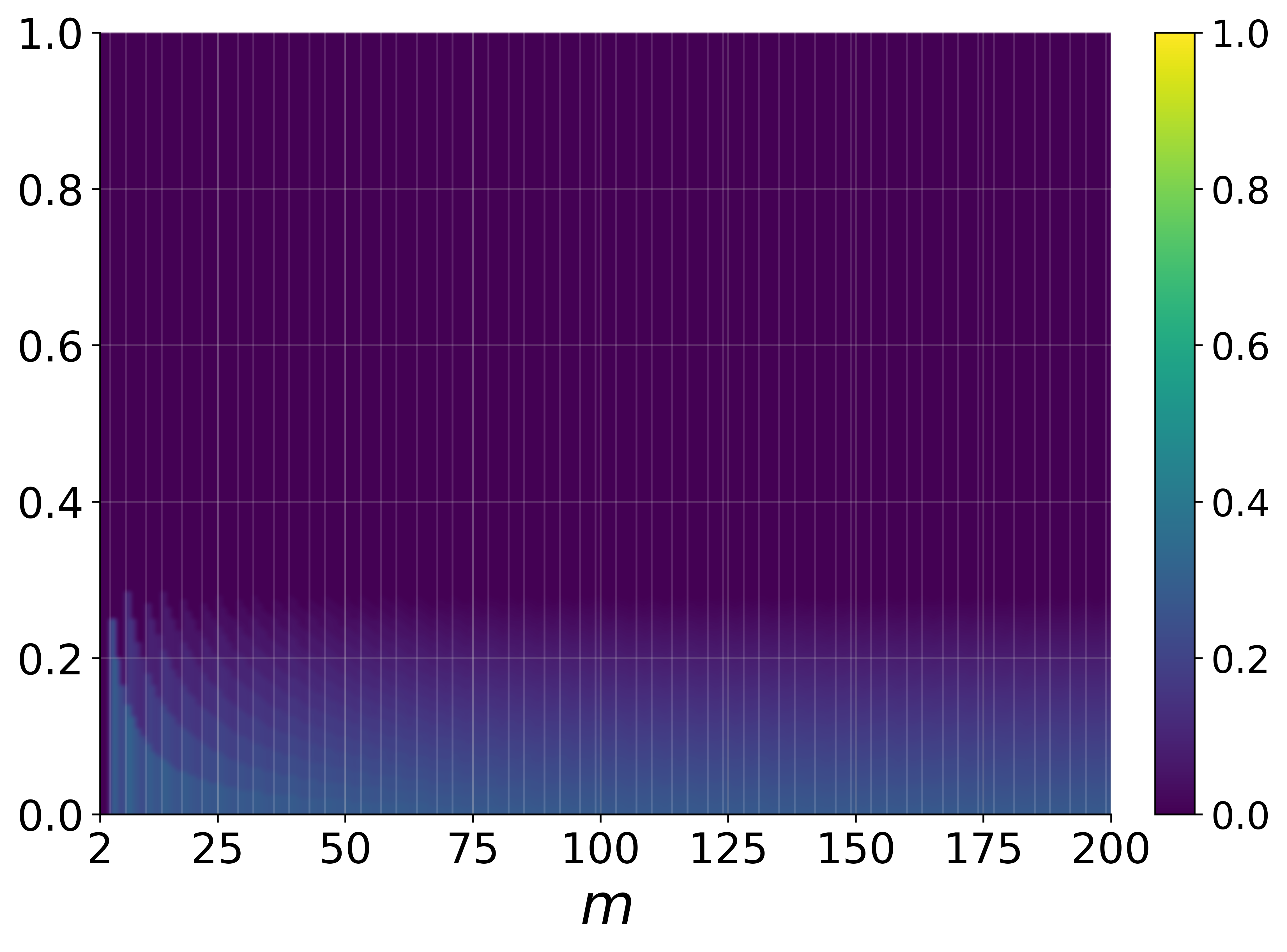}
        \caption{U-opt Distortion.}
        \label{fig:numtypes_dist_Uopt}
    \end{subfigure}\hfill
    \begin{subfigure}[t]{0.38\textwidth}
        \centering
        \includegraphics[width=\linewidth]{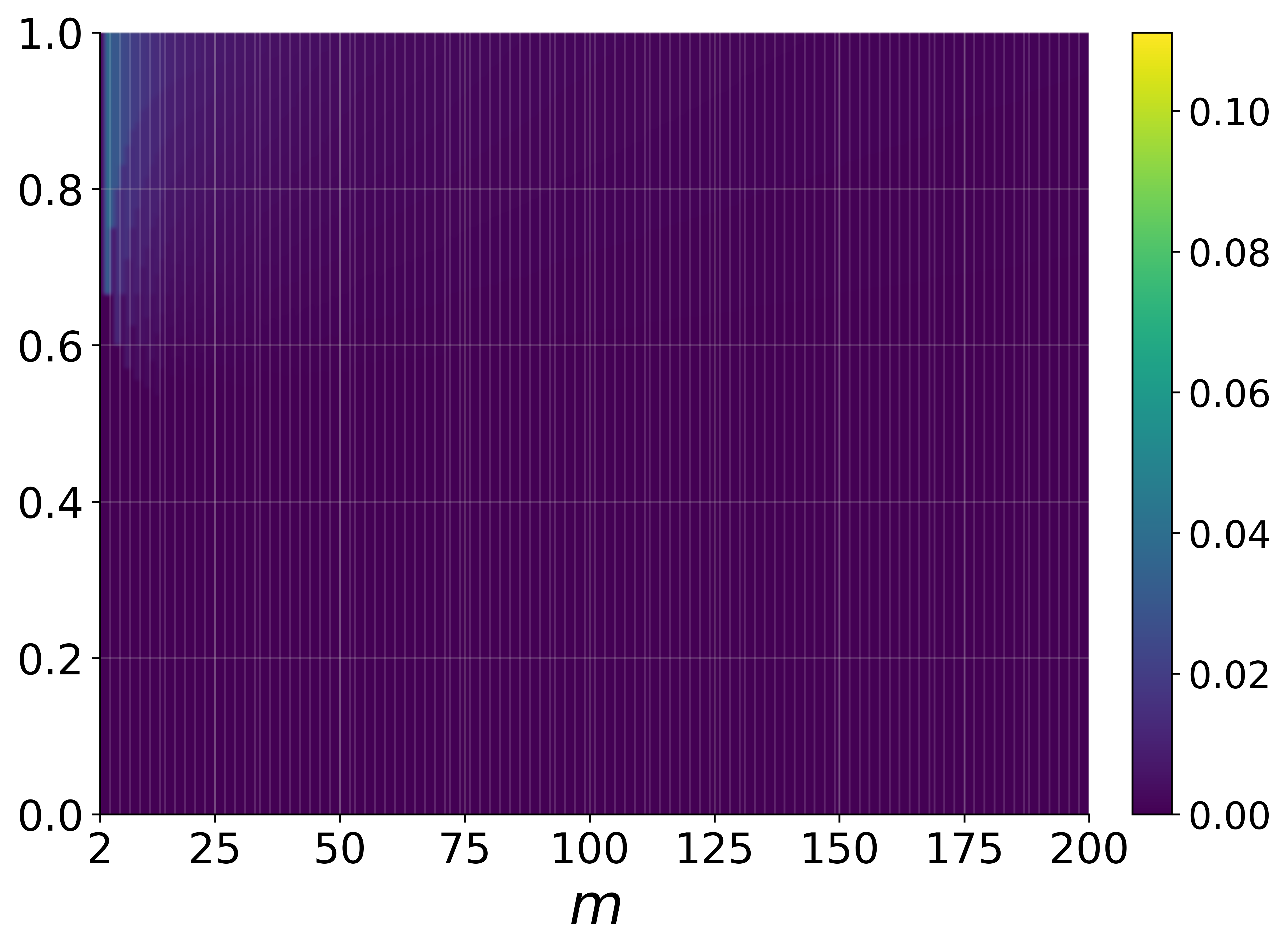}
        \caption{W-opt Audit probabilities.}
        \label{fig:numtypes_auditmass_Wopt}
    \end{subfigure}\hfill
    \begin{subfigure}[t]{0.38\textwidth}
        \centering
        \includegraphics[width=\linewidth]{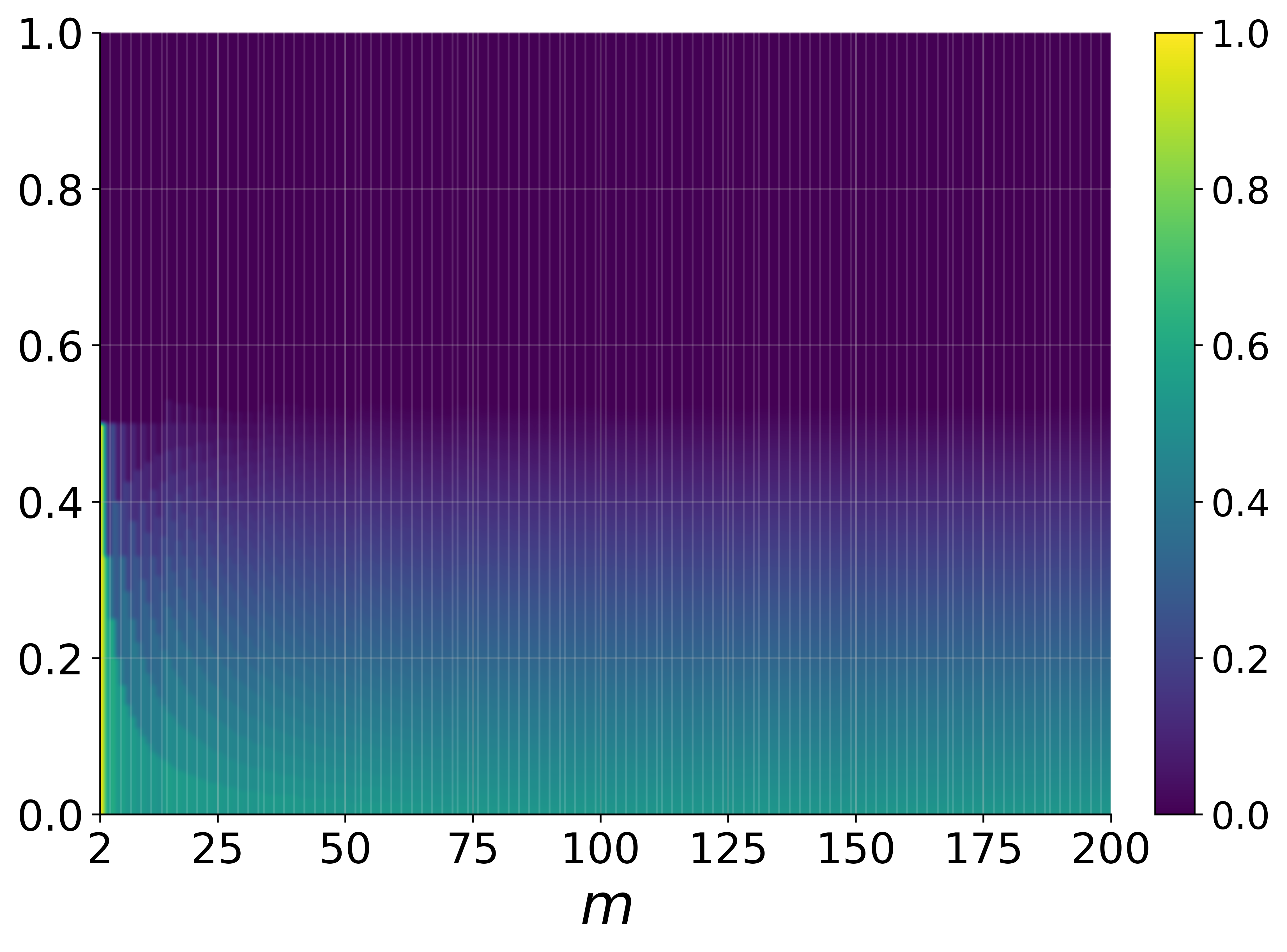}
        \caption{W-opt Distortion.}
        \label{fig:numtypes_dist_Wopt}
    \end{subfigure}

    \caption{Effect of resolution (number of types $m$). Top row: outcomes under the utility-optimal (solid) and welfare-optimal (dashed) policies. Bottom row: audit-mass and distortion heatmaps under the utility-optimal and welfare-optimal policies, respectively.}
    \label{fig:numtypes_8panels}
\end{figure*}

\Cref{fig:numtypes_8panels} shows that both objectives improve quickly as the number of types $m$ increases from very coarse discretizations, and then flatten out. \Cref{fig:numtypes_U} shows that the principal's utility under the utility-optimal policy rises rapidly for small $m$ and changes little once the type space is moderately fine. \Cref{fig:numtypes_W} shows the same pattern for social welfare under the welfare-optimal policy. 
The visible zigzag pattern is due to  discretization: when the grid becomes finer, new intermediate reports become available, which can change the equilibrium reporting pattern. For example, when the discretization is refined from $m=2$ to $m=4$, the lower half of the population is split into two finer types, and type $0$ may now pool to report $1$. This intermediate pattern cannot be represented in the two-type model, where that entire lower half is collapsed into a single type. These newly available report patterns can locally affect utility or welfare, which produces the small ups and downs across adjacent values of $m$. 

\Cref{fig:numtypes_mis,fig:numtypes_audit} show that the U-opt policy is generally stricter, as it sustains a higher audit rate and consequently reduces the mass of misreports relative to the W-opt policy.  Specifically, given the worst-case equilibrium $\mQ$, the misreporting mass is the fraction of misreported agents \(\sum_{i \ne k} q_i \mQ(i,k),\) and the audit rate is the fraction of audited agents\(\sum_{i,k} q_i \mQ(i,k) p_k\). \Cref{fig:numtypes_auditmass_Uopt,fig:numtypes_auditmass_Wopt} show similar patterns from the policy side: under both objectives, audit probabilities $p_k$ are larger for higher reports, while the U-opt policy assigns positive audit probability over a wider range of reports. For the distortion heatmaps, we plot the type-wise distortion \(\sum_k \mQ(i,k)\frac{|i-k|}{m-1}\) which is the normalized misreported distance.  \Cref{fig:numtypes_dist_Uopt,fig:numtypes_dist_Wopt} plot the distortion and show the complementary pattern on the equilibrium side.  The distortion is concentrated in a lower tail of types, and this distorted region is visibly thinner under the U-opt policy than under the W-opt policy.

\Cref{fig:numtypes_8panels} also shows that increasing the number of types mainly refines the discrete approximation of the same underlying environment. Once the resolution is moderately fine, further increases in $m$ have limited effect on utility or welfare, although a finer grid can still introduce new report patterns, which explains the remaining zigzag variation. The main qualitative difference between the two objectives remains the same throughout: the utility-optimal policy audits more aggressively and suppresses more misreporting, whereas the welfare-optimal policy saves audit cost and tolerates more distortion.

\begin{figure}[ht]
    \centering
    \begin{subfigure}{0.5\columnwidth}
        \includegraphics[width=0.8\columnwidth]{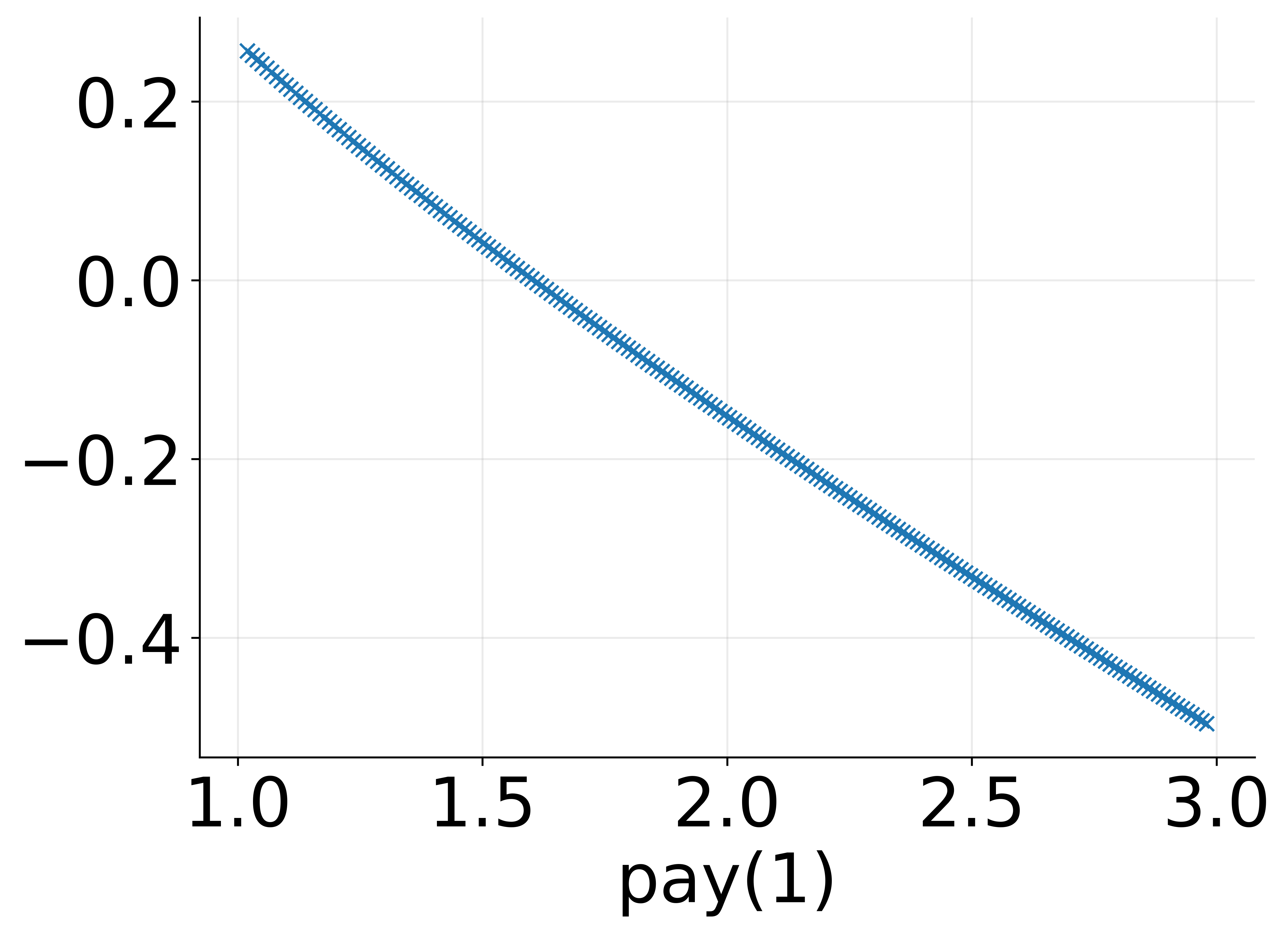}
        \caption{Principal's utility}\label{fig:pay1}
    \end{subfigure}~
    \begin{subfigure}{0.5\columnwidth}
        \includegraphics[width=0.8\columnwidth]{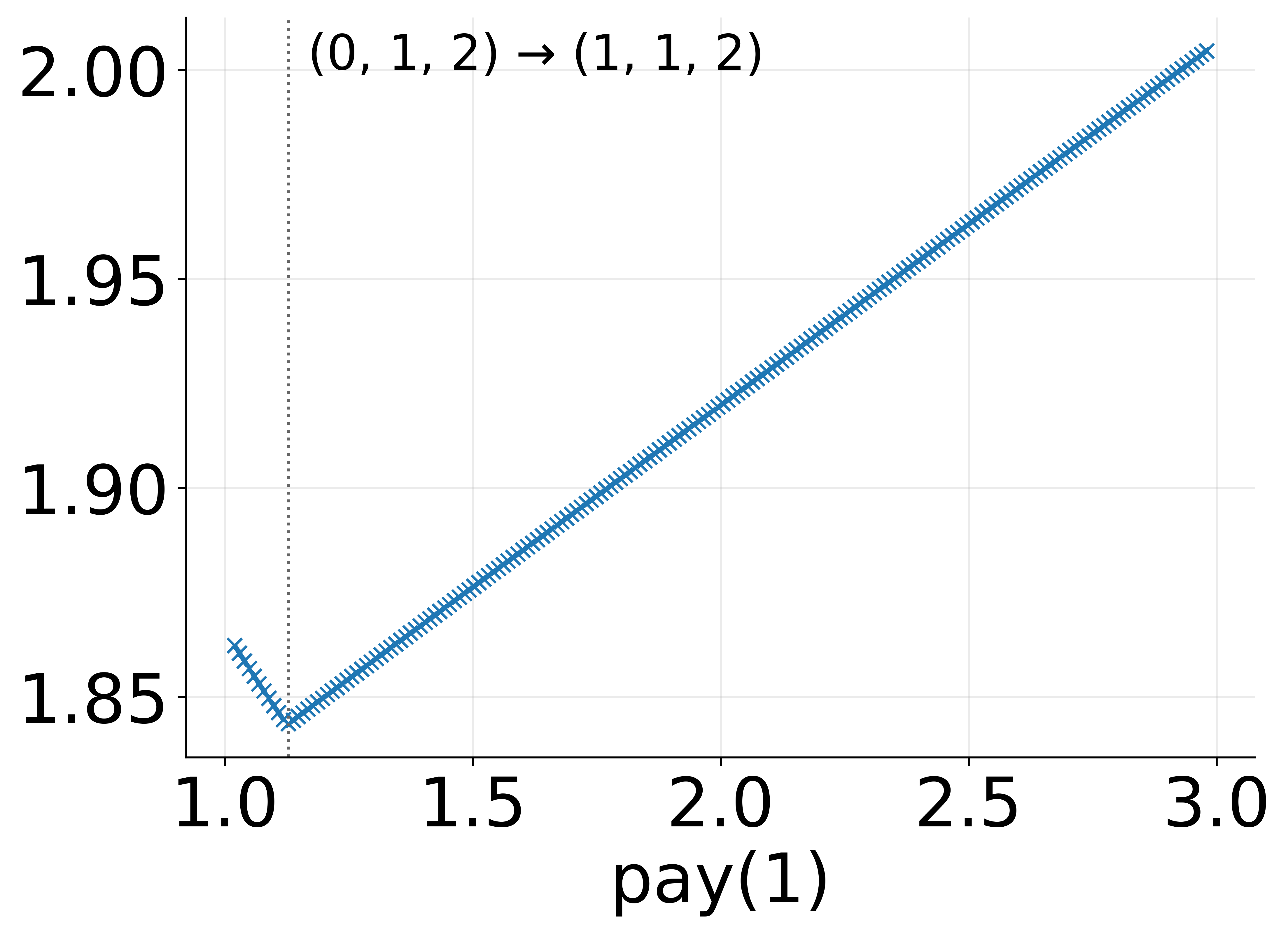}
        \caption{Social Welfare}\label{fig:pay2}
    \end{subfigure}
    \caption{Effect of $\pay$:  There are three types $m = 3$ and change the payment of type $1$ with the following parameters $n = 1$, $\vq = (0.4, 0.3, 0.3)$, $\val = \begin{pmatrix}
0.99 & 0.90 & 0.50 \\
0    & 1.50 & 1.40 \\
0    & 0    & 4.00
\end{pmatrix}$, $\pay(0) = 1, \pay(2) = 3$, $\pen = \pay + 0.5$, and $\lambda = 1$.
}\label{fig:pay}
\end{figure}



\section{Discussion and Future Work}
We provide several optimal and efficient audit policies for utility and welfare maximization under pessimistic equilibrium selection. At the same time, extending our model suggests fruitful directions for future work. First, extending our guarantees to finite agents, noisy or partial verification, and richer penalty structures remains open. Second, we take the classifier or allocation rule as exogenous; jointly designing the predictive model and the audit policy could yield better performance. Finally, it would be interesting to extend the incentive-minimization framework of \citet{estornell2023incentivizing} to non-binary payment outcomes.

\section*{Acknowledgments}
SD is grateful for support from NSF Award 2533162.

\bibliographystyle{plainnat}
\bibliography{references} 

\clearpage
\appendix

\crefalias{section}{appendix}

\section{Proofs in Section~\ref{sec:non-adaptive}}\label{app:non-adaptive}


\subsection{Proofs in Section~\ref{sec:non-adaptive_utility}}
\propnomax*
\begin{proof}[Proof of \cref{prop:nomax}]
Consider a binary type $m = 2$ setting.  First because the higher type never misreports to the lower type $Q_{1,0} = 0$, removing audit on the lower type, $p_0 = 0$ can never hurt the principal's utility, and we only need to consider audit vectors with $\vp = (0,p_1)$.  Given $\vp = (0,p_1)$, if $p_1>\rho_1(U_0)$ all agents are truthful and the principal's utility 
$$V_{tru}(p_1) = q_0 (\val(0,0)-\pay(0))+q_1 (\val(1,1)-\pay(1)-p_1\lambda).$$
If $p_1<\rho_1(U_0)$, the lower type agents always misreport as the higher type, and the utility 
$$V_{lie}(p_1) = q_0 (\val(0,1)-\pay(1)+p_1(\pen(1)-\lambda))+q_1 (\val(1,1)-\pay(1)-p_1\lambda).$$
Finally, if $p_1 = \rho_1(U_0)$, the lower type agents are indifferent between the higher or truth-telling and the principal's worst utility is $\min(V_{tru}(p_1), V_{lie}(p_1))$.  
Therefore, the principal's utility which can be discontinuous at $p_1 = \rho_1(U_0)$.  We provide one such example in~\Cref{fig:nomax}.
\end{proof}
\lembr*
\begin{proof}[Proof of \cref{lem:br}]
First, for all $i> \truthidx$ and $k\in [m]$ with $k\neq i$, because the truthful payment is higher than the highest misreport value $U_i>U_\truthidx\ge \hat{u}\ge \hat{U}_k$, truth-telling is the only best response $A_i = \{i\}$.  Second, for all $i<\truthidx$, and $k\in \hat{A}$, because $ U_i< \hat{u} = \hat{U}_k\le U_k$, type $i$ agents prefer $k$ over truth-telling and $i<k$.  Additionally, for all $l\in [m]$ with $l\neq i$ and $l\notin \hat{A}$, $U_{i,l} = \hat{U}_l<\hat{u} = U_{i,k}$.  Therefore, $A_i(\vp) = \hat{A}$.  The final case $\truthidx$ follows similarly.
\end{proof}

\lemequalwelldefined*
\begin{proof}[Proof of \cref{lem:equal_welldefined}]
To prove $\vp = \equal(u, A,\epsilon)\in [0,1]^m$, we only need to show that every coordinate $k\ge \iota$ is nonnegative because $\rho(u)\le 1$ for all $u\ge 0$. If $k\in A$, $p_k = \rho_k(u) = \frac{\pay(k)-u}{\pen(k)}\ge 0$ because $\pay(k)\ge\pay(\iota)\ge  u$.  If $k\ge \iota$ but $k\notin \hat{A}$, $p_k = \rho_k(u-\epsilon) = \frac{\pay(k)-u+\epsilon}{\pen(k)}\in[0,1]$ since $0<\epsilon< u$.  

For the second part, by \cref{eq:rho}, for all $k\ge \iota$, $\hat{U}_k(\vp) = u$ if $k\in \hat{A}$ and $\hat{U}_k(\vp) = u-\epsilon$ otherwise.  For all $k<\iota$, $\hat{U}_k\le U_k< u$.  Thus, $\hat{A}(\vp) = \argmax_k \{\hat{U}_k(\vp)\} = A$ and $\hat{u}(\vp) = \max_k \hat{U}_k(\vp) = u$.  $\truthidx(\vp) = \iota$ and $\vp$ is strict because $\hat{u}(\vp) = u$.
\end{proof}

\lemequalapprox*
\begin{proof}[Proof of \cref{lem:equal_approx}]
We consider three cases: strict audit policy with $\hat{u}(\vp)>0$, strict audit policy with $\hat{u}(\vp)\le 0$, and non-strict audit policy.  The idea of the first case is discussed in the main text. The other two cases require some modifications to ensure that the resulting policy is a strict equalized policy.

Given any strict audit policy $\vp\in [0,1]^m$ with $\hat{A}(\vp) = A$, $\hat{u}(\vp) = \hat{u}>0$, and $\truthidx(\vp) = \iota$ defined in \cref{lem:br}, we choose $\vp':= \equal(\hat{u}, \kappa, \epsilon)$ for an arbitrary $\kappa\in A$ and $0<\epsilon<\hat{u}$.  By \cref{lem:equal_welldefined}, $\vp'$ is strict and satisfies $\hat{A}(\vp') = \{\kappa\}\subseteq \hat{A}(\vp), \hat{u}(\vp') = \hat{u}$, and $\truthidx(\vp') = \iota$.  Therefore, the set of possible equilibria of the equalized one is a subset of the original policy's, $Eqi(\vp')\subseteq Eqi(\vp)$.  We now upper bound \cref{eq:equal_approx0} by upper bounding each term,
\begin{equation}\label{eq:equal_approx1}
    D_{k}:=\sum_{i}q_iQ_{i,k}(\pen(i,k)-\lambda)(p_k-p_k')
\end{equation}
For all $k\neq \kappa$ and $k< \iota$, because no one misreports or reports as $k$, $D_k = 0$.  
For all $k\neq \kappa$ and $k\ge \iota$, because no one misreports as $k$, $Q_{i,k} = 0$ for all $i\neq k$,
$D_k = q_k\lambda(p_k'-p_k).$  
Additionally, because 
\begin{align*}
    p_k'-p_k =& \rho_k(\hat{u}-\epsilon)-\rho_k(\hat{U}_k(\vp))\\
    =&\rho_k(\hat{u})+\frac{\epsilon}{\pen(k)}-\rho_k(\hat{U}_k(\vp))\\
    \le&\rho_k(\hat{u})+\frac{\epsilon}{\pen(k)}-\rho_k(\hat{u})\tag{$\hat{U}_k(\vp)\le \hat{u}$ and $\rho_k$ is decreasing}\\
    =&\frac{\epsilon}{\pen(k)}
\end{align*}
by \cref{eq:assum3} we have
$$D_k\le q_k\frac{\lambda\epsilon}{\pen(k)}\le q_k \epsilon.$$
Finally, because $p_\kappa' = \rho_\kappa(\hat{u}) = p_\kappa$, $D_\kappa = 0$.
Therefore, $V_\lambda(\vp, \mQ)-V_\lambda(\vp', \mQ) = n\sum_k D_k \le n\sum_{k\in [m]} q_k\epsilon = n\epsilon$.

If $\hat{u}(\vp)<0$, everyone is truthful $A_i(\vp) = \{i\}$.  We set $\vp':= \equal(\frac{U_0}{2}, \kappa, \epsilon)$ with $0<\epsilon<\frac{U_0}{2}$.  By \cref{lem:equal_welldefined}, $\truthidx(\vp') = 0$ and everyone is also truthful.  Hence, $Eqi(\vp') = Eqi(\vp) = \{\mathbb{I}\}$.  Because $p'_k\le \rho_k(0)<\rho_k(\hat{u}(\vp)) = p_k$, $D_k = q_k\lambda(p_k'-p_k)\le 0$ for all $k$. Therefore, $V_\lambda(\vp, \mQ)-V_\lambda(\vp', \mQ)\le 0$

For any non-strict audit policy $\vp\in [0,1]^m$ with $\hat{A}(\vp) = A$, $\hat{u}(\vp) = \hat{u}$, and $\truthidx(\vp) = \iota$, we choose $\vp':= \equal(\hat{u}-\frac{\epsilon}{2}, \kappa, \frac{\epsilon}{2})$ for an arbitrary $\kappa\in A$ and $0<\epsilon<\frac{\gamma}{2}$.  Because $U_\iota = \hat{u}>\hat{u}-\frac{\epsilon}{2}$, $\truthidx(\vp') = \iota$ and $\vp'$ is strict as $\epsilon$ is small enough.  Moreover, $\hat{A}(\vp') = \{\kappa\}$ and $\hat{u}(\vp') = \hat{u}-\frac{\epsilon}{2}$, so $Eqi(\vp')\subseteq Eqi(\vp)$.

Now we upper bound \cref{eq:equal_approx1}. 
For all $k<\iota$, $D_k = 0$ since no one misreports or reports as $k$. Second, if $k\neq\kappa$ and $i\ge \iota$, $D_{k} = q_k\lambda(p_k'-p_k) = q_k \lambda(\rho_k(\hat{u}-\epsilon)-\rho_k(\hat{U}_k(\vp)) \le q_k\epsilon$.  For $\kappa$, because $p_\kappa' = \rho_\kappa(\hat{u}-\frac{\epsilon}{2})$ and $p_\kappa = \rho_\kappa(\hat{u})$, we have $|p_\kappa'-p_\kappa|\le \frac{\epsilon}{\pen(\kappa)}$ and
$D_{\kappa} = q_\kappa \lambda (p_\kappa'-p_\kappa)+\sum_{i<\iota}q_i (\pen(\kappa)-\lambda)(p_\kappa-p_\kappa')\le \left(q_\kappa+\sum_{i<\iota}q_i\right)\epsilon.$
This completes the proof.
\end{proof}

\lemcriticalapprox*
\begin{proof}[Proof of \cref{lem:critical_approx}]
Similar to \cref{lem:equal_approx}, we prove they have the same set of equilibria, and use \cref{eq:equal_approx0} to show that the principal's utility is affine in $u$ under the same $\mQ$ so that extreme values are at the boundary. 

First, because $\vp$ is strict, $\truthidx(\vp) = \iota$, and $\epsilon'<\frac{\gamma}{2}$, $\hat{u}(\vp) = u, \hat{u}(\vp^-) = U_\iota-\epsilon'$ and $\hat{u}(\vp^+) = U_{\iota-1}+\epsilon'$ are in the interval $(U_{\iota-1}, U_{\iota})$.  Thus, by \cref{lem:equal_welldefined} $$A_i(\vp) = A_i(\vp^-) = A_i(\vp^+) = \begin{cases}
    \{i\}\text{ if }i\ge \iota\\
    \{\kappa\}\text{ otherwise}
\end{cases},$$
and $Eqi(\vp) = Eqi(\vp^-) = Eqi(\vp^+)$.

Second, we compute the difference of the principal's utility under $\mQ$.  Let $a_k:= n\sum_i q_iQ_{i,k}(\pen(i,k)-\lambda)$ for all $k\in [m]$ which is a constant independent of $\vp$ and $\delta^- = U_\iota-\epsilon'-u$.  By \cref{eq:equal_approx0}, 
\begin{align*}
    &V_\lambda(\vp^-, \mQ)-V_\lambda(\vp, \mQ)\\
    =& \sum_k a_k (p_k^--p_k)\\
    =& \left(\sum_{k<\iota} a_k (p_k^--p_k)\right)+a_\kappa (p_\kappa^--p_\kappa)+\left(\sum_{k>\iota: k\neq \kappa}a_k (p_k^--p_k)\right)\\
    =& a_\kappa (\rho_\kappa(U_\iota-\epsilon')-\rho_\kappa(u))+\sum_{k>\iota: k\neq \kappa}a_k (\rho_k(U_\iota-\epsilon'-\epsilon)-\rho_k(u-\epsilon))\tag{$p_k^- = p_k = 0$ for $k<\iota$}\\
    =& a_\kappa \frac{-\delta^-}{\pen(\kappa)}+\sum_{k>\iota: k\neq \kappa}a_k \frac{-\delta^-}{\pen(k)}\tag{$\rho_k(u+\delta)-\rho_k(u) = \frac{-\delta}{\pen(k)}$ for all $u,k,\delta$}\\
    =& \left(\sum_{k>\iota}\frac{-a_k}{\pen(k)}\right)\delta^-
\end{align*}
Therefore, let $\alpha = \left(\sum_{k>\iota}\frac{-a_k}{\pen(k)}\right)$.  We have
$$V_\lambda(\vp^-, \mQ)-V_\lambda(\vp, \mQ) = \alpha(U_\iota-\epsilon'-u)\text{ and }V_\lambda(\vp^+, \mQ)-V_\lambda(\vp, \mQ) = \alpha(U_{\iota-1}+\epsilon'-u).$$
If $\alpha\ge 0$, because $u<U_\iota$, $V_\lambda(\vp^-, \mQ)-V_\lambda(\vp, \mQ) \ge \alpha(U_\iota-\epsilon'-U_\iota) = -\alpha\epsilon'$.  If $\alpha<0$, because $u>U_{\iota-1}$, $V_\lambda(\vp^+, \mQ)-V_\lambda(\vp, \mQ) = |\alpha|(u-U_{\iota-1}-\epsilon')\ge -|\alpha| \epsilon'$.  Combining these two yields
$$V_\lambda(\vp)\le \max\{V_\lambda(\vp^-), V_\lambda(\vp^+)\}+|\alpha|\epsilon'.$$
Now we bound $|\alpha|\le n$.  
\begin{align*}
    \alpha =& \left(\sum_{k>\iota}\frac{-a_k}{\pen(k)}\right)\\
    =&\sum_{k>\iota}\frac{-1}{\pen(k)}n\sum_i q_iQ_{i,k}(\pen(i,k)-\lambda)\\
    =& n\sum_{k>\iota}\sum_i q_iQ_{i,k}\frac{-(\pen(i,k)-\lambda)}{\pen(k)}
\end{align*}
Because $|\pen(i,k)-\lambda|\le \max \{\lambda, \pen(k)-\lambda\}\le \pen(k)$, $|\alpha|\le n\sum_{k>\iota}\sum_i q_iQ_{i,k} \le n$.
\end{proof}

\begin{algorithm}[ht]
\caption{SuccinctSearch (fast evaluation over critical audit vectors)}
\label{alg:fast_search}
\begin{algorithmic}[1]
\Require $\epsilon>0$, $(n,m,q,\val,\pay,\pen)$, and $\lambda\ge 0$
\Ensure Audit vector $p^*$ maximizing $V_\lambda(p)$ over all critical audit vectors

\State $(F,H,S^{\mathrm{pay}},S^{\mathrm{inv}},M)\gets \Call{PrecomputeTables}{}$
\State $V_{\max}\gets -\infty$;\quad $(s^*,i^*,k^*)\gets (+,0,0)$

\For{$i\in [m]$}
    \For{$k=i$ to $m-1$}
        \For{$s\in\{+,-\}$}
            \State $v \gets \Call{ComputeValFast}{i,k,s,F,H,S^{\mathrm{pay}},S^{\mathrm{inv}},M}$
            \If{$v>V_{\max}$}
                \State $V_{\max}\gets v$;\quad $(s^*,i^*,k^*)\gets (s,i,k)$
            \EndIf
        \EndFor
    \EndFor
\EndFor

\State $\vp^* \gets \equal^{s^*}(i^*,k^*,\epsilon)$
\State \Return $p^*$

\Function{PrecomputeTables}{}
    \For{$i=0$ to $m$}
        \State $F[i] \gets \sum_{j<i} q_j$
        \State $H[i] \gets \sum_{j\ge i} q_j\cdot(\val(j,j)-\pay(j))$
        \State $S^{\mathrm{pay}}[i] \gets \sum_{j\ge i} q_j\cdot\frac{\pay(j)}{\pen(j)}$
        \State $S^{\mathrm{inv}}[i] \gets \sum_{j\ge i} q_j\cdot\frac{1}{\pen(j)}$
    \EndFor
    \For{$k=0$ to $m-1$}
        \State $M[0,k]\gets 0$
        \For{$i=1$ to $m$}
            \State $M[i,k]\gets M[i-1,k]+q_{i-1}\cdot \val(i-1,k)$
            \Comment{$M[i,k]=\sum_{j<i}q_j\val(j,k)$}
        \EndFor
    \EndFor
    \State \Return $(F,H,S^{\mathrm{pay}},S^{\mathrm{inv}},M)$
\EndFunction

\Function{ComputeValFast}{$i,k,s,F,H,S^{\mathrm{pay}},S^{\mathrm{inv}},M$}
    \Comment{Returns $V_\lambda(p)/n$ for the critical vector indexed by $(i,k,s)$; here $\hat{A}=\{k\}$.}
    \If{$s=+$}
        \State $c \gets \pay(i-1)$ \Comment{$\pay(-1):=0$}
    \Else
        \State $c \gets \pay(i)-2\epsilon$
    \EndIf
    \State $\hat{u} \gets c+\epsilon$
    \State $p_k \gets \dfrac{\pay(k)-\hat{u}}{\pen(k)}$
    \State $A_{\mathrm{truth}} \gets \Bigl(S^{\mathrm{pay}}[i]-c\cdot S^{\mathrm{inv}}[i]\Bigr)-q_k\cdot\dfrac{\epsilon}{\pen(k)}$
    \Comment{$A_{\mathrm{truth}}=\sum_{j\ge i}q_jp_j$ over truthful types}
    \State $v \gets M[i,k]-\pay(k)\cdot F[i]+(\pen(k)-\lambda)\,p_k\,F[i]+H[i]-\lambda\cdot A_{\mathrm{truth}}$
    \State \Return $v$
\EndFunction
\end{algorithmic}
\end{algorithm}

\subsection{Proofs in Section~\ref{sec:noregret}}



\lemcriticalnoregret*
\begin{proof}[Proof of \cref{lem:critical_noregret}]
Note that in the proof of \cref{lem:equal_approx}, given any audit vector $\vp\in [0,1]^m$, there exists a strict equalized audit vector $\vp' = \equal(u, \kappa, \epsilon)$ so that
$V_\lambda(\vp; \vq^t)\le V_\lambda(\vp'; \vq^t)+n\epsilon$ for all $\vq^t$, and hence
$$V_\lambda(\vp; \vec{\vq})\le V_\lambda(\vp'; \vec{\vq})+n\epsilon T.$$
By \cref{lem:critical_approx}, the principal's accumulative utility is still affine in $u$, so 
$$V_\lambda(\vp'; \vec{\vq})\le \max \{V_\lambda(\vp^+; \vec{\vq}), V_\lambda(\vp^-; \vec{\vq})\}+n\epsilon T$$
which completes the proof.
\end{proof}

\thmnoregret*
\begin{proof}[Proof of \cref{thm:noregret}]
    Let $V^t(\vp) = V_\lambda(\vp; \vq^t)$.  \Cref{alg:exp3} is basically EXP3~\cite{lattimore2020bandit} and uses $L$ to normalize reward to $[0,1]$, so the regret against best arm (critical vector) is 
    \begin{equation}\label{eq:exp31}
        \begin{aligned}
        &\max_{\sigma\in \Sigma}\sum_t V^t(\equal(\sigma, \epsilon_t))-\E_\mathcal{A}[\sum_t V^t(\vp^t)]\\
        \le &4nL\sqrt{T2m^2\log 2m^2} = O(n\sqrt{Tm^2\log m}).\end{aligned}
    \end{equation} 
    Moreover, by \cref{eq:equal_approx0}, for all $\epsilon>0$ and $\sigma\in \Sigma$, 
    \begin{equation}\label{eq:exp32}
       \left| \sum_t V^t(\equal(\sigma,\epsilon))-V^t(\equal(\sigma,\epsilon_t))\right| = O(n\sum_t \epsilon_t+nT\epsilon)
    \end{equation}
    where $\sum_t \epsilon_t\le 2\epsilon_0$ is a constant.
    Therefore, with \cref{lem:critical_noregret}, \cref{eq:exp31,eq:exp32}, and $\epsilon\to 0$, we have 
     $$ \sup_{\vp}\sum_t V^t(\vp)-\E_\mathcal{A}[\sum_t V^t(\vp^t)]  = O(n\sqrt{Tm^2\log m}+2n\epsilon_0)
     $$
     which completes the proof 
\end{proof}
\subsection{Proofs in Sections~\ref{sec:non-adaptive_welfare} and \ref{sec:non-adaptive_penality}}
\propmonotone*
\begin{proof}[Proof of \cref{prop:monotone}]
First note that increasing audit cost can only decrease the principal's utility and social welfare, 
   $V_\lambda(\vp; \pen) \leq V_{\lambda'}(\vp; \pen)$ and $W_\lambda(\vp; \pen)\le W_{\lambda'}(\vp, \pen)$.  To complete the proof, it is sufficient to prove there exists some $\vp'$ so that 
   $V_\lambda(\vp; \pen) \leq V_{\lambda}(\vp'; \pen')$ and $W_\lambda(\vp; \pen)\le W_{\lambda}(\vp', \pen')$.

Given $\vp, \pen$ and $\pen'$, we define a new policy $\vp'$ so that for all $k\in [m]$
$$p_k' = \frac{\pen(k)}{\pen'(k)} p_k$$
where $0\le p_k'\le p_k$ for all $k$ because $0<\pen(k)\le\pen'(k)$. Note that any type-$i$ agent's  utility of reporting as $k\neq i$ under the new penalty and audit policy $\vp'$ is 
$\pay(k)-p_k'\pen'(k)= \pay(k)-\frac{\pen(k)}{\pen'(k)} p_k\pen'(k) = \pay(k)-p_k\pen(k)$ which is equal to the utility under the original penalty and audit policy $\vp$, so $Eqi(\vp, \pen) = Eqi(\vp', \pen')$.  Finally because $p_k'\le p_k$ for all $k$, the difference of utility~\cref{eq:util_principal} is 
\begin{align*}
    &V_\lambda\bigl(\vp,\pen\bigr)-V_\lambda\bigl(\vp',\pen'\bigr)\\
    =& n\sum_{i,k} q_iQ_{i,k} [p_k(\pen(i,k)-\lambda)-p_k'(\pen'(i,k)-\lambda)]\\
    =& -n\sum_{i,k} q_iQ_{i,k} \lambda(p_k-p_k')\le 0.
\end{align*}

Similarly, the difference of welfare in~\cref{eq:welfare_max} is 
\begin{align*}
    W_\lambda(\vp, \pen)-W_\lambda(\vp',\pen') =& -n\sum_{i,k} q_iQ_{i,k} [p_k\lambda-p_k'\lambda] \le 0.
\end{align*}
This shows that raising each $\pen(k)$ can only maintain or increase the principal’s worst-case utility.
\end{proof}

\section{Proofs of Section~\ref{sec:adaptive}}\label{app:adaptive}
\lemdict*
\begin{proof}[Proof of \cref{lem:dict}]
We consider three cases as in~\cref{eq:dict}.  By definition, if $\mQ\in Eqi(\vp^*)$ and $\hat{\vq} = \hat{\vq}^*$, $\mQ\in Eqi(\pi_{dict})$.  Second, if $\mQ$ has $\hat{\vq}\neq \hat{\vq}^*$ and $\hat{\vq}\neq \vq$, $\mQ$ is not truthful.  Since $\pi_{dict}(\hat{\vq}) = \mathbf{1}$ audits everyone, any misreporting agent would deviate to truthful reporting by \cref{eq:assum1}, and $\mQ$ cannot be an equilibrium.  Finally, if $\mQ$ has $\hat{\vq}\neq \hat{\vq}^*$ and $\hat{\vq} = \vq$, not all agents (mis)report as the highest type because $\vq$ has full support on $[m]$.   Because $\pi_{dict}(\hat{\vq}) = \mathbf{0}$ does not audit anyone, all agents would like to report as the highest possible type, so $\mQ$ cannot be an equilibrium.
\end{proof}

\lembesteqi*
\begin{proof}[Proof of \cref{lem:best_eqi}]
We first prove for strict audit vectors with $\hat{u}(\vp)>0 $ and consider the single-minded strategy $\mQ'$ with  $\iota = \truthidx(\vp)$ and $\kappa = \min_i\hat{A}(\vp)$.  Note that if $\hat{u}(\vp)\le 0$, $Eqi(\vp) = \{\mathbb{I}\}$ is also single-minded.  By~\cref{lem:br}, $\mQ'$ is an equilibrium under $\vp$.  Given $\vp$, let
$V_{i,k} := \val(i,k)-\pay(k)+p_k\left(\pen(i,k)-\lambda\right)$ for all $i,k$ be the contribution of an agent of type $i$ reporting as $k$ to the principal's utility. For any $\mQ\in Eqi(\vp)$, as
$$V_\lambda(\vp,\mQ) = \sum_{i}q_i\sum_{k\in A_i(\vp)} Q_{i,k} V_{i,k},$$ the utility depends on convex combinations of $V_{i,k}$ over $k\in A_i$ for each $i$.  Additionally, the utility of single-minded strategy with $\iota, \kappa$ is 
    \begin{equation}\label{eq:best_eqi2}
        V_\lambda(\vp,\mQ') = \sum_{i<\iota}q_i V_{i,\kappa}+\sum_{i\ge \iota}q_i V_{i,i}
    \end{equation}
    For all $i<\truthidx$ and $k<l\in A_i = \hat{A}(\vp)$, 
    \begin{align*}
        &V_{i,k}-V_{i,l}\\
        =& \left(\val(i,k)-\pay(k)+p_k\left(\pen(k)-\lambda\right)\right)-\left(\val(i,l)-\pay(l)+p_l\left(\pen(l)-\lambda\right)\right)\tag{$i\neq k,l$}\\
        =& \val(i,k)-\val(i,l)-(p_k-p_l)\lambda-[\pay(k)-p_k\pen(k)-\pay(l)+p_l\pen(l)]\\
        =& \val(i,k)-\val(i,l)-(p_k-p_l)\lambda\tag{$U_{i,k}= U_{i,l}$}
    \end{align*}
    First, by \cref{eq:assum2} $\val(i,k)-\val(i,l)\ge 0$.  Now we show $p_l-p_k\ge 0$.  Because $\frac{\pay(l)-u}{\pay(k)-u}\ge \frac{\pay(l)}{\pay(k)}>\frac{\pen(l)}{\pen(k)}$ for all $0< u\le \pay(k)$ by \cref{eq:assum4} and $\pay(l)>\pay(k)$, $\rho_l(u)>\rho_k(u)$.  Hence, 
    \begin{equation}\label{eq:best_eqi0}
        p_l = \rho_l(\hat{u}(\vp))> p_k = \rho_k(\hat{u}(\vp))\text{ for all } k\le l\in \hat{A}(\vp)
    \end{equation}
    As a result, for all $l\in \hat{A}(\vp)$
    \begin{equation}\label{eq:best_eqi1}
        V_{i,\kappa}\ge V_{i,l}.
    \end{equation}
    Therefore, 
    \begin{align*}
        &V_\lambda(\vp,\mQ)\\
        =& \sum_{i<\iota}\sum_k q_i Q_{i,k}V_{i,k}+\sum_{i\ge\iota}\sum_k q_i Q_{i,k}V_{i,k}\\
        =&  \sum_{i<\iota} \sum_k q_i Q_{i,k}V_{i,k}+\sum_{i\ge\iota} q_i V_{i,i}\tag{type $i\ge \iota$ is truthful}\\
        \le& \sum_{i<\iota} \sum_k q_i Q_{i,k}V_{i,\kappa}+\sum_{i\ge\iota} q_i V_{i,i}\tag{by \cref{eq:best_eqi1}}\\
        =& \sum_{i<\iota} q_i V_{i,\kappa}+\sum_{i\ge\iota} q_i V_{i,i} = V_\lambda(\vp,\mQ').\tag{by \cref{eq:best_eqi2}}
    \end{align*}
If $\vp$ is not strict and $\truthidx(\vp)$ is indifferent between truth-telling and misreporting as $\kappa = \min_i\hat{A}(\vp)$, the above argument implies that the best equilibrium will be either $\truthidx$ reporting as $\kappa$ or truthfully, which are both single-minded equilibria.  
\end{proof}

\thmoptadaptivecostly*
\begin{proof}[Proof of \cref{thm:opt_adaptive_costly}]
It is sufficient to find an algorithm that maximizes \cref{eq:adaptive_up},
$$\sup_{\vp, \mQ\in Eqi(\vp)} V_\lambda(\vp, \mQ) = \max_{\mQ\in Eqi}\sup_{\vp: \mQ\in Eqi(\vp)}V_\lambda(\vp,\mQ)$$
By \cref{lem:best_eqi}, we can go through all single-minded report strategies and find an audit policy $\vp$ that supports such strategy as an equilibrium and optimizes the utility.  First, as each single-minded strategy has two parameters $i, k$, there are $O(m^2)$ candidates.  For any single-minded strategy $\mQ$ with $i,k$ and $\vp$ with $\mQ\in Eqi(\vp)$, there exists a strict equalized policy $\vp' = \equal(i,k,\epsilon)$ so that $\{\mQ\} = Eqi(\vp')$ by \cref{lem:equal_welldefined}. By \cref{lem:equal_approx}
$V_\lambda(\vp', \mQ) \ge V_\lambda(\vp, \mQ)-n\epsilon$ and, by \cref{lem:critical_approx}, there exists a critical policy $\vp''\in \{\equal^-(i,k,\epsilon),\equal^+(i,k,\epsilon)\}$ so that 
$V_\lambda(\vp'',\mQ)\ge V_\lambda(\vp',\mQ)-n\epsilon$.  
Therefore, for any single-minded strategy $\mQ$, there exists a critical policy $\vp''$---$\equal^-(i,k,\epsilon)$ or $\equal^+(i,k,\epsilon)$---so that
$$V_\lambda(\vp'',\mQ)\ge \sup_{\vp:\mQ\in Eqi(\mQ)}V_\lambda(\vp,\mQ)-2n\epsilon.$$
Because each critical policy only has one equilibrium which is single-minded, \FVal returns $V_\lambda(\vp'') = V_\lambda(\vp'', \mQ)$ and thus, \cref{alg:water_filling_algorithm} can be seen as going through all single-minded equilibria and can return $\vp^*$ and $\mQ^*$ so that
$$V_\lambda(\vp^*,\mQ^*)\ge \max_{\mQ}\sup_{\vp:\mQ\in Eqi(\mQ)}V_\lambda(\vp,\mQ)-2n\epsilon.$$

Finally, by \cref{lem:dict}, the dictator audit strategy $\pi^*$ with $\vp^*$ and $\hat{\vq}^*$ associated with $\mQ^*$ achieves 
\begin{equation}\label{opt_adaptive_costly0}
    V_\lambda(\pi^*) = \min_{\mQ\in Eqi(\pi^*)}V_\lambda(\vp^*,\mQ)\ge \min_{\mQ\in Eqi(\vp^*)}V_\lambda(\vp^*,\mQ) = V_\lambda(\vp^*,\mQ^*)
\end{equation}

The upper bound $\sup_{\vp, \mQ\in Eqi(\vp)} V_\lambda(\vp, \mQ) \le \max_{\mQ\in Eqi}\sup_{\vp: \mQ\in Eqi(\vp)}V_\lambda(\vp,\mQ)$ is trivial. The above argument provides a dictator strategy that achieves $\sup_{\vp,\mQ: \mQ\in Eqi(\vp)}V_\lambda(\vp,\mQ)-2n\epsilon$ for all $\epsilon>0$ which completes the other direction.
\end{proof}

\section{Proofs in Section~\ref{sec:adaptive_budgeted}}\label{app:budget}

\lemsmallbudget*
\begin{proof}[Proof of \cref{lem:small_budget}]
Given $\pi\in \Pi_B$, let $\vp$ be the output of $\pi$ on $\hat{\vq} = (0,\dots,0,1)$ which is the report distribution when all report as $m-1$.  We show that all reporting to $m-1$ is an equilibrium.  As $\pi\in \Pi_B$ and $\hat{q}_{m-1} = 1$, $p_{m-1}\le \frac{B}{n \hat{q}_{m-1}} = \frac{B}{n}$.  Hence, for all $k<m-1$
\begin{align*}
        &\hat{U}_{m-1}-{U}_{k}\\
        =& U_{m-1}-p_{m-1}\pen(m-1)-U_{k}\\
        \ge& U_{m-1}-U_{k}-\frac{B}{n\hat{q}_{m-1}}\pen(m-1)\tag{$p_{m-1}\le \frac{B}{n \hat{q}_{m-1}}$}\\
        \ge& U_{m-1}-U_{k}-(U_{m-1}-U_{m-2})\tag{$\frac{B}{n}\le  \frac{U_{m-1}-U_{m-2}}{\pen(m-1)}$}\\
        =& U_{m-2}-U_k\ge 0\tag{$\hat{q}_{m-1}\le 1$}
\end{align*}
Thus, all prefer reporting as $m-1$ to truth-telling.  Since $\hat{U}_k\le U_k$, $m-1\in A_i$ for all $i\in [m]$, and all reporting to $m-1$ is an equilibrium.

The upper bound is trivial, because $p_{m-1}\le \frac{B}{n}$.
\end{proof}

\lemdictbudgeted*
\begin{proof}[Proof of \cref{lem:dict_budgeted}]
We consider three types of report strategies.  If $\mQ\in Eqi(\vp^*)$ and $\hat{\vq} = \hat{\vq}^*$, $\mQ\in Eqi(\pi_{dict})$.  Second, if $\mQ$ has $\hat{\vq}\neq \hat{\vq}^*$ and $\hat{q}_{m-1}>q_{m-1}$, some agent misreports as the highest type $m-1$.  Similar to \cref{lem:small_budget}, we can show that reporting as the second highest type has a higher utility than reporting as $m-1$.  If $p_{m-1} = \frac{B}{n\hat{q}_{m-1}}$, 
    \begin{align*}
        &\hat{U}_{m-1}-\hat{U}_{m-2}\\
        =& U_{m-1}-p_{m-1}\pen(m-1)-U_{m-2}\tag{$p_{m-2} = 0$}\\
        <& U_{m-1}-U_{m-2}-\frac{n\beta}{n\hat{q}_{m-1}}\pen(m-1)\\
        =& U_{m-1}-U_{m-2}-\frac{1}{\hat{q}_{m-1}}\pen(m-1)\frac{U_{m-1}-U_{m-2}}{\pen(m-1)}\\
        =& (U_{m-1}-U_{m-2})(1-1/\hat{q}_{m-1})\le 0.
    \end{align*}  On the other hand, if $p_{m-1} = 1$, $\hat{U}_{m-1}-\hat{U}_{m-2} < U_{m-1}-U_{m-2}-\pen(m-1)\le 0$ by \cref{eq:assum1}.  
    Similarly, $\hat{U}_{m-1}<U_{m-2}$, so $\mQ$ cannot be an equilibrium.

For the final case, not everyone reports $m-1$, and they will deviate to reporting $m-1$, which is a contradiction.

Finally, because the latter two cases always satisfy the budget constraints, $\pi_{dict}\in \Pi_B$ if $\vp^*\in \Pi_B(\hat{\vq}^*)$.
\end{proof}

\lembesteqibudgeted*
\begin{proof}[Proof of \cref{lem:best_eqi_budgeted}]

By \cref{lem:br,lem:equal_welldefined}, $\mQ'$ is an equilibrium under $\vp'$.  For the budget constraint, we have
\begin{align*}
    &C(\vp,\mQ) = n\sum_{i,k}q_i Q_{i,k} p_k\\
    =& n\sum_{i\ge \iota} q_i p_i+n\sum_{i< \iota}\sum_{k\in \hat{A}}q_i Q_{i,k} p_k\\
    \ge& n\sum_{i\ge \iota} q_i p_i+n\sum_{i< \iota}q_i p_\kappa\tag{$\sum_{k\in \hat{A}}Q_{i,k} = 1$ and $p_\kappa\le p_k$ by \cref{eq:best_eqi0}}\\
    \ge& n\sum_{i\ge \iota} q_i p_i'+n\sum_{i< \iota}q_i p_\kappa'\tag{$p_i\ge p_i'$ for all $i\ge \iota$}\\
    =& C(\vp,\mQ').
\end{align*}   On the other hand, 
\begin{align*}
    &V(\vp,\mQ)\\
    =&\sum_{i<\iota}\sum_{k\in \hat{A}} q_i Q_{i,k}(\val(i,k)-\pay(k)+p_k\pen(k))+\sum_{i\ge\iota}q_i (\val(i,i)-\pay(k))\tag{type $i\ge \iota$ is truthful}\\
        \le& \sum_{i<\iota} q_i (\val(i,\kappa)-\pay(\kappa)+p_\kappa\pen(\kappa))+\sum_{i\ge\iota} q_i (\val(i,i)-\pay(k))\tag{$\hat{U}_k = \hat{U}_\kappa$ and \cref{eq:assum2}}\\
        =& V(\vp',\mQ').\tag{$p_\kappa' = p_\kappa$}
\end{align*} 
\end{proof}

\lembestequal*
\begin{proof}[Proof of \cref{lem:best_equal}]

Given \cref{eq:opt_adaptive_budgeted1}, note that the cost $C(\vp', \mQ)$ is a function on $u$.
$$\begin{aligned}
    &C(u):=C(\vp',\mQ) = n\sum_{i,k} q_iQ_{i,k}p_k' \\
    = &n\sum_{i<\iota} q_i p_\kappa'+\sum_{i\ge\iota} q_ip_i' = n\sum_{i<\iota} q_i \rho_\kappa(u)+\sum_{i\ge\iota} q_i\rho_i(u).
\end{aligned}$$
Moreover, because $C(u)$ is a decreasing affine function in $u$ by \cref{eq:rho}, we set 
\begin{equation}
u = \max\left\{U_{\iota-1}, C^{-1}(B)\right\}\label{eq:best_equal0}
\end{equation}
which can be solved in $O(m)$.  Since $C$ is decreasing, 
\begin{equation}\label{eq:best_equal1}
    C(u) = C(\max\left\{U_{\iota-1}, C^{-1}(B)\right\})\le \min \{C(U_{\iota-1}), B\}\le B
\end{equation}

For any $\vp\in \Pi_B(\mQ)$, by \cref{lem:br}, 
$U_{\iota-1}\le \hat{u}\le U_\iota$ where $\hat{u}=\hat{u}(\vp)$.  
Because
\begin{align*}
    &B\ge C(\vp,\mQ)\\
    =& n\sum_{i<\iota}q_ip_\kappa+\sum_{i\ge \iota}q_ip_i\\
    \ge& n\sum_{i<\iota}q_i\rho_\kappa(\hat{u})+\sum_{i\ge \iota}q_i\rho_i(\hat{u})\tag{$\rho_i$ are decreasing}\\
    =& C(\hat{u})
\end{align*}
and $C$ is decreasing, $C^{-1}(B)\le \hat{u}$.  Combining these two yields, 
\begin{equation}\label{eq:best_equal2}
    U_{\iota-1}\le u\le \hat{u}(\vp)\le U_\iota.
\end{equation}  
Thus $\mQ\in Eqi(\vp')$ by \cref{lem:br}, and $\vp'\in \Pi_B(\mQ)$ by \cref{eq:best_equal1}.  Finally, by \cref{eq:best_equal2}, $p_\kappa = \rho_\kappa(\hat{u}(\vp))\le \rho_\kappa(u) = p_\kappa'$, so $V(\vp',\mQ)-V(\vp,\mQ)= n\sum_{i<\iota}q_i Q_{i,\kappa}(p_\kappa'-p_\kappa)\pen(\kappa)\ge 0.$
\end{proof}

\clearpage

\end{document}